\documentclass{article}
\usepackage{arxiv}
\usepackage{graphicx}%
\usepackage{multirow}%
\usepackage{amsmath,amssymb,amsfonts}%
\usepackage{amsthm}%
\usepackage{mathrsfs}%
\usepackage[title]{appendix}%
\usepackage{xcolor}%
\usepackage{textcomp}%
\usepackage{manyfoot}%
\usepackage{booktabs}%
\usepackage{algorithmicx}%
\usepackage{algpseudocode}%
\usepackage{listings}%
\usepackage{orcidlink}

\newtheorem{theorem}{Theorem}
\newtheorem{proposition}[theorem]{Proposition}%

\usepackage{stmaryrd}
\usepackage{tikz}
\usetikzlibrary{automata, positioning, arrows, trees}
\tikzset{%
    ->,  
    >=stealth, 
    node distance=2.5cm, 
    every state/.style={thick, fill=gray!10}, 
    initial text=$ $, 
}
\usepackage{stackengine}
\usepackage{wrapfig}
\usepackage[linesnumbered]{algorithm2e}

\newtheorem{Proposition}[theorem]{Proposition}
\newtheorem{Example}{Example}
\newtheorem{Definition}{Definition}
\newtheorem{Corollary}[theorem]{Corollary}
\newtheorem{Lemma}[theorem]{Lemma}
\newtheorem{Theorem}[theorem]{Theorem}

\usepackage{savesym}
\savesymbol{U}
\usepackage[mathscr]{euscript}
\usepackage{dsfont}
\usepackage{tikz-cd}

\newcommand{\powerset}[1]{\mathscr{P}(#1)}

\newcommand{\Lstar}{\ensuremath{\mathrm{L^\ast}}}

\newcommand{\Nat}{\ensuremath{\mathbb{N}}}

\newcommand{\Tol}{\ensuremath{\mathcal{S}}}

\newcommand{\RED}{\ensuremath{\mathbf{RED}}}
\newcommand{\BLUE}{\ensuremath{\mathbf{BLUE}}}

\newcommand{\MQ}{\ensuremath{\mathbf{MQ}}}

\newcommand{\ra}{\ensuremath{\rightarrow}}
\newcommand{\indist}{\ensuremath{\equiv}}

\newcommand{\Real}{\ensuremath{\mathbb{R}}}

\newcommand{\Aut}{\ensuremath{A}}
\newcommand{\AutB}{\ensuremath{B}}
\newcommand{\Clique}{\ensuremath{C}}

\newcommand{\quanti}{\ensuremath{I}}

\newcommand{\Sta}{\ensuremath{Q}}
\newcommand{\Func}{\ensuremath{F}}
\newcommand{\rankseq}{\ensuremath{R}}
\newcommand{\NTree}{\ensuremath{T}}
\newcommand{\Wor}{\ensuremath{W}}
\renewcommand{\S}{\ensuremath{S}}

\newcommand{\func}{\ensuremath{f}}
\newcommand{\sta}{\ensuremath{q}}

\newcommand{\worw}{\ensuremath{w}}
\newcommand{\woru}{\ensuremath{u}}
\newcommand{\worv}{\ensuremath{v}}
\newcommand{\tol}{\ensuremath{t}}
\newcommand{\simil}{\ensuremath{{z}}}

\newcommand{\rankk}[1]{\ensuremath{R_{#1}}}

\newcommand{\Symb}{\ensuremath{\Sigma}}
\newcommand{\Psimplex}{\ensuremath{\Delta}}

\newcommand{\symb}{\ensuremath{\sigma}}
\newcommand{\policy}{\ensuremath{\pi}}
\newcommand{\tra}{\ensuremath{\tau}}
\newcommand{\emptyW}{\ensuremath{\lambda}}
\newcommand{\quantp}{\ensuremath{\kappa}}
\newcommand{\pdist}{\ensuremath{\delta}}
\newcommand{\strtosta}{\ensuremath{\beta}}

\newcommand{\PLstar}{\ensuremath{\mathrm{L_{\Stasim}^\ast}}}

\newcommand{\cla}[1]{\ensuremath{\llbracket #1 \rrbracket}}
\newcommand{\ccla}[1]{\ensuremath{\left[#1\right]}}

\newcommand{\eqdef}{\triangleq}

\newcommand{\clique}{\ensuremath{c}}

\newcommand{\Stasim}{\ensuremath{\mathcal{E}}}

\newcommand{\lmodel}{\ensuremath{\mathcal{M}}}

\newcommand{\qindist}[1]{\ensuremath{\overline{#1}}}

\newcommand{\terminal}{\ensuremath{\$}}
\newcommand{\SymbT}{\ensuremath{\Symb_\terminal}}
\newcommand{\staI}{\ensuremath{\sta_{\mathrm{in}}}}
\newcommand{\ProbT}{\ensuremath{\Psimplex(\SymbT)}}
\newcommand{\Words}{\ensuremath{\Symb^\ast}}
\newcommand{\policyW}{\ensuremath{\policy^\ast}}

\newcommand{\traW}{\ensuremath{\tra^\ast}}

\newcommand{\EQ}{\ensuremath{\mathbf{EQ}}}

\newcommand{\Pref}{\ensuremath{\mathit{Pre}}}
\newcommand{\Suff}{\ensuremath{\mathit{Suf}}}

\newcommand{\hstaI}{\ensuremath{\overline{\sta}_{\mathrm{in}}}}
\newcommand{\htraW}{\ensuremath{\overline{\tra}^\ast}}
\newcommand{\hpolicyW}{\ensuremath{\overline{\policy}^\ast}}
\newcommand{\hpolicy}{\ensuremath{\overline{\policy}}}
\newcommand{\htra}{\ensuremath{\overline{\tra}}}

\newcommand{\hSta}{\ensuremath{\overline{\Sta}}}
\newcommand{\hAut}{\ensuremath{\overline{\Aut}}}

\newcommand{\bstaI}{\ensuremath{\overline{\sta}_{\mathrm{in}}}}
\newcommand{\btraW}{\ensuremath{\overline{\tra}^\ast}}
\newcommand{\bpolicyW}{\ensuremath{\overline{\policy}^\ast}}
\newcommand{\bpolicy}{\ensuremath{\overline{\policy}}}
\newcommand{\btra}{\ensuremath{\overline{\tra}}}

\newcommand{\bSta}{\ensuremath{\overline{\Sta}}}

\newcommand{\qstaI}{\qindist{\sta}_{\mathrm{in}}}
\newcommand{\qtraW}{\qindist{\tra}^\ast}
\newcommand{\qpolicyW}{\qindist{\policy}^\ast}
\newcommand{\qpolicy}{\qindist{\policy}}
\newcommand{\qtra}{\qindist{\tra}}
\newcommand{\qsta}{\qindist{\sta}}
\newcommand{\qSta}{\qindist{\Sta}}

\newcommand{\qFunc}{\qindist{\Func}}
\newcommand{\qfunc}{\qindist{\func}}
\newcommand{\qpdfa}{\ensuremath{H}}

\newcommand{\tsim}{\ensuremath{\approx}}
\newcommand{\wsim}{\approxeq}

\newcommand{\Realp}{\ensuremath{\Real_{+}}}
\newcommand{\iffdef}{\ensuremath{\stackrel{\vartriangle}{\iff}}}

\renewcommand{\setminus}{\backslash}

\newcommand{\topr}[1]{\ensuremath{\mathit{top}_{#1}}}
\newcommand{\rank}[1]{\ensuremath{\mathit{rank}_{#1}}}
\newcommand{\rthr}{\ensuremath{r}}

\newcommand{\supsim}{\ensuremath{\mathit{sdr}}}
\newcommand{\supeq}{\ensuremath{\mathit{supeq}}}
\newcommand{\wer}[1]{\ensuremath{\mathit{wer}_{#1}}}
\newcommand{\ndcg}[1]{\ensuremath{\mathit{ndcg}_{#1}}}
\newcommand{\vd}{\ensuremath{\mathit{vd}}}

\DeclareMathOperator{\supp}{supp}
\DeclareMathOperator{\CG}{CG}
\DeclareMathOperator{\DCG}{DCG}
\DeclareMathOperator{\NDCG}{NDCG}

\newcommand{\Equ}{\Stasim}
\newcommand{\qedexample}{\ensuremath{\blacksquare}}

\newcommand{\ce}{\worv}

\newcommand{\card}{\ensuremath{\#}}

\newcommand{\Close}{\textsf{Close}}
\newcommand{\Consistent}{\textsf{Consistent}}

\newcommand{\BuildQPDFA}{\textsf{BuildQPDFA}}
\newcommand{\Update}{\textsf{Update}}
\newcommand{\True}{\textsc{True}}
\newcommand{\prefixes}{\textsf{prefixes}}
\newcommand{\suffixes}{\textsf{suffixes}}

\newcommand{\qlmodel}{\ensuremath{\qindist{\lmodel}}}

\newcommand{\qfuncW}{\ensuremath{\qfunc^\ast}}

\begin{document}

\title{Congruence-based Learning of Probabilistic Deterministic Finite Automata}
\renewcommand{\shorttitle}{Congruence-based Learning of PDFA}

\author{%
    {M. Carrasco \hspace*{1.5ex}
    F. Mayr \hspace*{1.5ex}
    S. Yovine}\\
    Facultad de Ingeniería\\ Universidad ORT Uruguay\\ Montevideo, Uruguay\\
    \texttt{carrasco\_m@ort.edu.uy}\\
    \texttt{mayr@ort.edu.uy}\\
    \texttt{yovine@ort.edu.uy}\\
}

\renewcommand{\headeright}{Carrasco, Mayr, Yovine}

\date{}

\maketitle

\begin{abstract}
This work studies the question of learning probabilistic deterministic automata from language models. 
For this purpose, it focuses on analyzing the relations defined on algebraic structures over strings by equivalences and similarities on probability distributions. We introduce a congruence that extends the classical Myhill-Nerode congruence for formal languages. This new congruence is the basis for defining regularity over language models. We present an active learning algorithm that computes the quotient with respect to this congruence whenever the language model is regular. The paper also defines the notion of recognizability for language models and shows that it coincides with regularity for congruences. For relations which are not congruences, it shows that this is not the case. Finally, it discusses the impact of this result on learning in the context of language models.
\end{abstract}




\section{Introduction}\label{sec:intro}

In the last few years, there has been a growing interest in trying to understand sequence processing neural networks via capturing their behavior with finite automata through active learning by means of adapting Angluin's \Lstar\ learning algorithm~\cite{Angluin:1987}. 
Works like~\cite{WeissGY24, mayr_yovine_2021, bollig_23} deal with neural binary classifiers of finite sequences over finite alphabets with the aim to learning deterministic finite automata (DFA). 
For neural language models, the goal is to learn a probabilistic deterministic finite automaton (PDFA)~\cite{IEEE:Vidal2005}. In this case, two approaches have been studied: 
those which view neural networks as producing the probability of the input sequence~\cite{sico_pdfa_distillation_24, pmlr-v217-muskardin23a},
and those which consider a network to be an autoregressive model that outputs the next-symbol probability distribution~\cite{weiss_WFA_learning, icgi2023_a}. 
This paper focuses on the latter.

The algorithm proposed in~\cite{weiss_WFA_learning} is based on a tolerance relation induced by the supremum distance between probability distributions in order to group states which have similar futures, in the sense that when continued by the same sequence they reach states that remain close to each other. This is achieved by clustering the observations obtained through querying the target neural language model with so-called \emph{membership queries}. An important drawback of this approach is that the non-transitivity of the tolerance inherited from the distance implies that the clusters are not unique. 
In contrast, the learning algorithm developed in~\cite{icgi2023_a} steps on equivalences over probability distributions in order to define a family of congruences over the set of sequences. The advantage of this is that a congruence being a transitive tolerance, it induces a unique partition of the set of observations got by membership queries, which allows avoiding the possibly arbitrary grouping choices made by the clustering method.
Indeed, this approach is aligned with the one behind \Lstar\ whose cornerstone is the relation between regular languages and DFA induced by the Myhill-Nerode congruence. 

The contribution of this paper is two-fold. 
First, it studies the mathematical properties of the tolerance and congruence structures induced on the set of sequences by languages models for any similarity or equivalence relation over probability distributions. 
Second, it steps on these properties to analyze the learning capabilities of these approaches provided the kind of relation they are based on. 

To achieve this, the paper starts by reviewing similarities and equivalences commonly used in the literature for the analysis of language models (Section~\ref{sec:pdist}).
It then shows that these similarities and equivalence relations induce tolerances and congruences over sequences, respectively, it proves several results about them, and it defines the quotient given by a congruence and the notion of regular language model (Section~\ref{sec:lmodels}).  
Section~\ref{sec:pdfa} recalls the formal definition of PDFA, it introduces the concept of quotient PDFA, and it proves several results establishing the formal relations between language models, PDFA and their quotients.
Section~\ref{sec:plstar} proposes the algorithm \PLstar\ which generalizes Angluin's \Lstar\ algorithm with the purpose of learning a quotient PDFA from a language model modulo an equivalence over distributions \Equ. Correctness and termination of \PLstar\ heavily rely on the properties of quotients proved before. It also defines the concept of recognizability of language models and proves that it coincides with regularity in the case of congruences.
Section~\ref{sec:tolerance} is devoted to the analysis of the feasibility of developing learning algorithms that rely on tolerances as opposed to congruences. The main result is that recognizability does not longer imply regularity. The impact of this fact in learning is discussed.
Section~\ref{sec:conclusions} summarizes the contributions.

\section{Probability distributions}\label{sec:pdist}

Let \Symb\ be a finite \emph{alphabet} and $\SymbT \eqdef \Symb \cup \{\terminal\}$, where \terminal\ is a special \emph{terminal} symbol not in \Symb. 
A \emph{probability distribution} over \SymbT\ is a function $\pdist:\SymbT\ra[0,1]$ such that $\sum_{\symb\in\SymbT} \pdist(\symb)=1$.
We denote \ProbT\ the set of all probability distributions over \SymbT.

\subsection{Similarities}

We call \emph{similarity}, denoted \Tol, a reflexive and symmetric binary relation between distributions over \SymbT. We write $\pdist\tsim_\Tol\pdist'$ to indicate that $\pdist$ and $\pdist'$ are related by \Tol.

A natural way of defining a similarity is by means of a function $\simil:\ProbT\times\ProbT\to\Real_+$ satisfying:
\begin{enumerate}
    \item \simil\ is symmetric in the sense that $\simil(\pdist,\pdist')=\simil(\pdist',\pdist)$ for all $\pdist,\pdist'\in\ProbT$;
    \item and $\simil(\pdist,\pdist)=0$ for all $\pdist\in\ProbT$.
\end{enumerate} 
Such a function \simil\ induces a similarity via a threshold $\tol\in\Realp$ denoted $\tsim_{(\simil,\tol)}$ and defined as:
\begin{align}\label{def:simil}
\pdist \tsim_{(\simil,\tol)} \pdist &\iffdef \simil(\pdist,\pdist')\leq\tol
\end{align}
By choosing the function \simil\ one gets different examples of similarities appearing in the literature.
Before listing some of them, let us define $\rank{}(\pdist) : \SymbT \ra \Nat$ to be the \emph{ranking} of symbols $\symb\in\SymbT$ induced by their probability $\pdist(\symb)$:
\begin{align}\label{def:rank}
\rank{}(\pdist)(\symb) &\eqdef \card\{\pdist(\symb')\mid\pdist(\symb')\geq\pdist(\symb),\,\symb'\in\SymbT\}.   
\end{align}
We assume $\rank{}(\pdist)$ to be injective, or equivalently, that there are no ties. We can achieve this assuming $\SymbT$ to be equipped with an arbitrary ordering and break ranking ties using this ordering. For $\rthr\in\Nat$, $\topr{\rthr}(\pdist) \subseteq \SymbT$, gives the set of top-\rthr\ ranked symbols while forgetting their relative order: 
\begin{align}\label{def:top}
\topr{\rthr}(\pdist) &\eqdef \{ \symb\in\SymbT \mid \rank{}(\pdist)(\symb) \leq \rthr\}
\end{align}
The following are examples of similarities which are of interest in our context.
\begin{description}
    \item[Variation Distance] corresponds to the choice of the infinity norm for \simil: 
    \begin{align}\label{def:tolerance}
    \vd(\pdist,\pdist') &\eqdef \max_{\symb\in\SymbT} |\pdist(\symb)-\pdist'(\symb)|
    \end{align}
    This relation was used in \cite{weiss_WFA_learning, clark_thollard}. \\
    
    \item[Support Difference Rate] A basic example of similarity is obtained by taking
    \begin{align}
    \supsim(\pdist,\pdist') &\eqdef \frac{\card\big(\supp(\pdist)\oplus\supp(\pdist')\big)}{\card\SymbT}\in[0,1]        
    \end{align}
    where $\card$ denotes cardinal, $\oplus$ the symmetric difference between sets, and $\supp$ the \emph{support} of the distribution (set of elements with non-zero probability). The function \supsim\ measures the proportion of symbols belonging to the support of one distribution but not to the other. \\
    
    \item[Word Error Rate] For $\rthr\in\Nat$, consider the function:
    \begin{align}\label{def:wer}
    \wer{\rthr}(\pdist,\pdist') &\eqdef \frac{1}{2\rthr}\card\big(\topr{\rthr}(\pdist)\oplus\topr{\rthr}(\pdist')\big)\in[0,1]
    \end{align}
    We obtain in this way the word error rate (WER), which measures the proportion of symbols which are $\topr{\rthr}$ for one distribution but not for the other. 
    This relation was used in~\cite{weiss_WFA_learning} with $\rthr=1$. \\
    
    \item[Cumulative Gain] The previous example can also be obtained in a slightly different way by considering the cumulative gain with respect to a binary relevance measure. Fix $\rthr\in\Nat$ and define the relevance of symbol $\symb$ with respect to $\pdist$ as $\mathds{1}\big[\symb\in\topr{\rthr}(\pdist)\big]$. Then the cumulative gain of $\pdist'$ with respect to $\pdist$ for this relevance measure is given by:
    \begin{align}
    \CG_{\rthr}\left(\pdist'\mid \pdist\right) &\eqdef \sum_{\symb\in \topr{\rthr}(\pdist')} \mathds{1}\big[\symb\in\topr{\rthr}(\pdist)\big]        
    \end{align} 
    Since $\max_{\pdist'}\CG_{\rthr}\left(\pdist'\mid \pdist\right)=r$, the normalized cumulative gain is $\frac{1}{r}\CG_{\rthr}$. Then the previous example (Def.~\ref{def:wer}) can be rewritten as:
    \begin{align}
    \wer{\rthr}(\pdist,\pdist') &\eqdef 1 - \frac{\CG_{\rthr}\left(\pdist'\mid \pdist\right) + \CG_{\rthr}\left(\pdist\mid \pdist'\right)}{2\rthr}
    \end{align}
    Notice that $\CG_{\rthr}\left(\pdist'\mid \pdist\right)$ coincides with the cardinal of the intersection of the respective $\topr{\rthr}$ symbols, which is symmetric. This symmetry breaks down in the general case of an arbitrary relevance measure or when considering a discounted cumulative gain as in the next example. \\

    \item[Discounted Cumulative Gain] Let us denote by $\symb^{\pdist}_{k}$ the symbol whose ranking with respect to $\pdist$ equals $k$. That is $\rank{}(\pdist)\left(\symb_{k}^{\pdist}\right)=k$. Also, let $\rank{\rthr}$ be the function: 
    \begin{equation}\label{def:rankr_fun}
     \rank{\rthr}(\pdist) \eqdef
     \begin{cases}
         \rank{}(\pdist) & \text{if }\rank{}(\pdist)\leq\rthr\\
         r+1 & \text{otherwise}.
     \end{cases}   
    \end{equation}
    In order to account for the specific ranking of symbols, we can modify the previous example by considering a measure of relevance related to the rank together with a discounted cumulative gain: 
    \begin{align}
    \DCG_{\rthr}\left(\pdist'\mid \pdist\right) &\eqdef \sum_{k=1}^r \frac{1}{\log_2(k+1)} \left[r-\rank{\rthr}(\pdist)\left(\symb_{k}^{\pdist'}\right) + 1\right]        
    \end{align}
    The discounting factor weights each ranking position so that the symbols that mostly account for the total sum are those appearing at the top of the ranking.
    Since the maximum of $\DCG_{\rthr}\left(\pdist'\mid \pdist\right)$ is attained when $\pdist'=\pdist$, the normalized discounted cumulative gain is given by:
    \begin{align}\label{eq:NDCG}
    \NDCG_{\rthr}\left(\pdist'\mid \pdist\right) &\eqdef \left.\DCG_{\rthr}\left(\pdist'\mid \pdist\right)
    \middle/
    \sum_{k=1}^r \frac{r-k+1}{\log_2(k+1)} \right.     
    \end{align}
    Therefore we can consider the function:
    \begin{align}\label{def:ndcg}
    \ndcg{\rthr}(\pdist,\pdist') \eqdef 1 - \frac{\NDCG_{\rthr}\left(\pdist'\mid \pdist\right) + \NDCG_{\rthr}\left(\pdist\mid \pdist'\right)}{2}        
    \end{align}
    More generally, we can define the relevance of symbol $\symb$ with respect to a distribution $\pdist$ as $\pdist(\symb)$, and define the respective normalized discounted cumulative gain.
    Such version of this similarity was used in~\cite{pmlr-v57-balle16}.
\end{description}

\subsection{Equivalences}

Let \Stasim\ be an \emph{equivalence} relation between distributions over \SymbT. 
For simplicity, we write $\pdist =_{\Stasim} \pdist'$ to denote that $(\pdist, \pdist')\in\Stasim$.
We denote by $\ccla{\pdist}_\Equ$ the class of $\pdist$ and $\ccla{\cdot}_\Equ$ the quotient map. 
Several examples of equivalences are of interest. The following relations were used in~\cite{learnaut2022, icgi2023_a, icgi2023_b, learnaut2024}.
\begin{description} 
\item[Quantization] 
Given a \emph{quantization parameter} $\quantp\in\Nat$, $\quantp\geq 1$, the quantization \emph{interval} $\quanti_\quantp^n$, for $n\in\Nat$, $0\leq n<\quantp-1$, is the interval $\left[ n \quantp^{-1}, (n+1) \quantp^{-1} \right)$, and for $n=\quantp-1$, is the interval $\left[ n \quantp^{-1}, 1 \right]$. 
For $\pdist, \pdist'\in\ProbT$:
\begin{align}\label{def:quant}
\pdist =_\quantp \pdist' 
&\iffdef \textrm{for all}\ \symb\in\SymbT\ .\
(\pdist(\symb), \pdist'(\symb))\in\quanti_\quantp^n\times\quanti_\quantp^n\ 
\textrm{for some}\ n    
\end{align}
\item[Rank]
For $\rthr\in\Nat$, and given distributions $\pdist$ and $\pdist'$, we define
\begin{align}\label{def:rankr}
\pdist =_{\rank{\rthr}} \pdist' &\iffdef \rank{\rthr}(\pdist) = \rank{\rthr}(\pdist')
\end{align}
\item[Top] 
For $\rthr\in\Nat$ and $\pdist, \pdist'\in\ProbT$:
\begin{align}\label{def:topr}
\pdist =_{\topr{\rthr}} \pdist' &\iffdef \topr{\rthr}(\pdist) = \topr{\rthr}(\pdist') 
\end{align}
\item[Support] 
For $\pdist, \pdist'\in\ProbT$:
\begin{align}\label{def:supp}
\pdist =_{\supeq} \pdist' &\iffdef \supp(\pdist) = \supp(\pdist') 
\end{align}
\item[Combinations] One can combine two equivalence relations $\Stasim_1$ and $\Stasim_2$ by defining $\pdist =_{\Stasim} \pdist'$ if and only if $\pdist =_{\Stasim_i} \pdist'$ for $i=1,2$. For example, $=_\quantp$ and $\topr{\rthr}$ combined with $\supeq$ were used for analyzing constrained language models in~\cite{learnaut2024}.
\end{description}

\subsection{Properties}

Intuitively, the examples of similarities and equivalences defined above suggest that they are related. We show here several results that formalize their relationship. \\

\begin{proposition}\label{prop:quantp_vs_tol}
For every $\pdist, \pdist' \in \Psimplex(\Symb_\terminal)$, if $\pdist=_\quantp\pdist'$ then $\pdist\tsim_{(\vd,\quantp^{-1})}\pdist'$.
\end{proposition}
\begin{proof}
By Def.~\ref{def:quant}, $\pdist=_\quantp\pdist'$ implies for all $\symb\in\SymbT$, 
$(\pdist(\symb), \pdist'(\symb)) \in \quanti_\quantp^n \times \quanti_\quantp^n$ for some $n$. 
Thus, by definition of $\quanti_\quantp^n$, $|\pdist(\symb)-\pdist'(\symb)| \leq \quantp^{-1}$ for all $\symb\in\SymbT$. 
So, $\max_{\symb\in\SymbT} |\pdist(\symb)-\pdist'(\symb)| \leq \quantp^{-1}$.
Hence, by Def.~\ref{def:tolerance}, $\pdist\tsim_{(\vd,\quantp^{-1})}\pdist'$.
\qedhere
\end{proof}

~\\
In general, it is not true that $\tsim_{(\simil,\tol)}$ is an equivalence relation for $\tol>0$.
Nevertheless, when the function $\simil$ is a pseudometric, meaning that it satisfies the triangle inequality $\simil(\pdist,\pdist'')\leq \simil(\pdist,\pdist')+\simil(\pdist',\pdist'')$, the relation $\tsim_{(\simil,0)}$ with threshold $\tol=0$ is indeed an equivalence.
Clearly, this is the case for \vd\ and \supsim, which become $=$ and \supeq, respectively.
Moreover, the equivalence relations $\topr{\rthr}$ and $\rank{\rthr}$ are related to the evaluation metrics WER and NDCG, sometimes used to compare language models \cite{pmlr-v57-balle16, weiss_WFA_learning}. This is made precise in the following propositions. \\

\begin{Proposition}\label{prop:top_vs_wer}
For every $\pdist, \pdist' \in \Psimplex(\Symb_\terminal)$, 
\(\pdist =_{\topr{\rthr}} \pdist'\) if and only if \(\wer{\rthr}(\pdist,\pdist')=0\).
\end{Proposition}
\begin{proof}
From Def.~\ref{def:wer} we have that \(\wer{\rthr}(\pdist,\pdist')=0\) is equivalent to $\topr{\rthr}(\pdist)=\topr{\rthr}(\pdist')$.
\end{proof}

\begin{Proposition}\label{prop:rank_vs_ndcg}
For every $\pdist, \pdist' \in \Psimplex(\Symb_\terminal)$ we have \(\pdist =_{\rank{\rthr}} \pdist'\) if and only if \(\ndcg{\rthr}(\pdist,\pdist')=0\).
\end{Proposition}
\begin{proof}
Suppose first that $\rank{\rthr}(\pdist')=\rank{\rthr}(\pdist)$. Then $\rank{\rthr}(\pdist)\left(\symb_k^{\delta'}\right)=k$ since $\symb_k^{\pdist'}=\symb_k^{\pdist}$. Therefore by Def.~\ref{eq:NDCG} we have $\NDCG_{\rthr}\left(\pdist'\mid \pdist\right)=1$. Analogously we have $\NDCG_{\rthr}\left(\pdist\mid \pdist'\right)=1$. Thus, by Def.~\ref{def:ndcg}, $\ndcg{\rthr}(\pdist,\pdist')=0$.

Let us prove that $\ndcg{\rthr}(\pdist,\pdist')=0$ implies $\rank{\rthr}(\pdist')=\rank{\rthr}(\pdist)$. This amounts to showing that the function
\begin{align}\label{eq:rankproperty}
\S(\rankseq) &\eqdef \sum_{k=1}^r\frac{r-\rankk{k}+1}{\log _2(k+1)}
\end{align}
defined over all sequences $\rankseq=\left(\rankk{k}\right)_{k=1}^r$ of integers in $\{1,2, \ldots, r, r+1\}$ with no repetitions except (possibly) for $r+1$, has a unique maximum at $\rankk{k}=k$. To prove this claim, first notice that if $\rankk{k}=r+1$ for some $k$, then $r-\rankk{k}+1=0$ and there is no contribution to the sum in Eq.~\ref{eq:rankproperty}. So we can assume that $\rankseq$ is a permutation of $\{1,2, \ldots, r\}$. Suppose that $\rankk{i}>\rankk{j}$ for some indices $i<j$. Since $\log _2(i+1)<\log _2(j+1)$, swapping $\rankk{i}$ and $\rankk{j}$ yields a higher sum. Indeed, if $\rankseq'$ is the sequence with $\rankk{i}$ and $\rankk{j}$ swapped, then
\[
\S(\rankseq')-\S(\rankseq)
=
\left(\rankk{i}-\rankk{j}\right) \left[\frac{1}{\log_2(i+1)}-\frac{1}{\log_2(j+1)}\right]>0.
\]
Therefore the maximum must satisfy $\rankk{i}<\rankk{j}$ if $i<j$. That is $\rankk{k}=k$ for all $k=1,\ldots,r$.
\end{proof}
\section{Language models}\label{sec:lmodels}

A \emph{language model} is a total function $\lmodel: \Words \ra \ProbT$ that maps every string in \Words\ to a probability distribution over \SymbT, where $\lmodel(\woru)(\symb)$ is the probability of \woru\ to be continued by symbol \symb.

Given a language model \lmodel\ and a similarity relation $\Tol$ on $\ProbT$, we define the relation $\wsim^\lmodel_{\Tol} \subseteq \Words\times\Words$ as follows:
\begin{align}\label{def:simil_str}
\woru \wsim^\lmodel_{\Tol} \woru' &\iffdef 
\forall \worw\in\Words.\ \lmodel(\woru\worw) \tsim_{\Tol} \lmodel(\woru'\worw)
\end{align}
Actually, $\wsim^\lmodel_{\Tol}$ is a \emph{tolerance} relation~\cite{tolerance76} on the algebraic structure $(\Words, \Func)$, where $\Func \eqdef \{ \func_\symb:\Words\ra\Words \mid \symb\in\Symb\}$ such that $\func_\symb(\woru) \eqdef \woru \symb$, i.e., $\func_\symb$ appends symbol \symb\ to strings. That is, $\wsim_{\Tol}$ is a reflexive, symmetric and compatible relation:\\

\begin{Proposition}\label{prop:tolerance_wsim}
$\wsim^\lmodel_{\Tol}$ is a \emph{tolerance} relation on $(\Words, \Func)$.
\end{Proposition}
\begin{proof}
\begin{description}
    \item[Reflexivity] Let $\woru\in\Words$:
    \begin{align*}
        \lmodel(\woru) \tsim_{\Tol} \lmodel(\woru) 
        &\implies 
        \woru \wsim_{\Tol} \woru
        & \textrm{by reflexivity of } \tsim_{\Tol}
    \end{align*}
    \item[Symmetry] Let $\woru,\woru'\in\Words$: 
    \begin{align*}
        \woru \wsim^\lmodel_{\Tol} \woru'
        &\implies \forall\worw\in\Words.\ 
            \lmodel(\woru\worw) \tsim_{\Tol} \lmodel(\woru'\worw) 
            & \textrm{by Def.\ref{def:simil_str}} \\
        &\implies \forall\worw\in\Words.\ 
            \lmodel(\woru'\worw) \tsim_{\Tol} \lmodel(\woru\worw) 
            & \textrm{by symmetry of } \tsim_{\Tol} \\
        &\implies \woru' \wsim^\lmodel_{\Tol} \woru      
            & \textrm{by Def.\ref{def:simil_str}}
    \end{align*}
    \item[Compatibility] Let $\woru,\woru'\in\Words$ such that $\woru \wsim^\lmodel_{\Tol} \woru'$, and $\symb\in\Symb$. 
    Then, for all $\worw\in\Words$:
    \begin{align*}
        \lmodel((\woru\symb)\worw) 
        &= \lmodel(\woru(\symb\worw)) \\
        &\tsim_{\Tol} \lmodel(\woru'(\symb\worw)) 
            & \textrm{by hypothesis and Def.~\ref{def:simil_str}}      \\
        &= \lmodel((\woru'\symb)\worw)
    \end{align*}
    Hence, $\woru\symb \wsim^\lmodel_{\Tol} \woru'\symb$. \qedhere   
\end{description}
\end{proof}

~\\
The relation $\tsim_{\Tol}$ induces a reflexive and symmetric relation $\wsim_{\Tol}$ between language models as follows:
\begin{align}\label{def:tolerance_lm}
\lmodel_1 \wsim_{\Tol} \lmodel_2 
&\iffdef \forall\woru\in\Words.\ 
\lmodel_1(\woru) \tsim_{\Tol} \lmodel_2(\woru)
\end{align}
However, $\lmodel_1\wsim_{\Tol}\lmodel_2$ does not imply that $\wsim^{\lmodel_1}_{\Tol}$ and $\wsim^{\lmodel_2}_{\Tol}$ are the same tolerance relation over \Words. 
We will illustrate this with an example.\\

\begin{Example}\label{ex:diff_tolerances}
Consider the alphabet $\Symb=\{a\}$ and the language models defined on $\Words$ by
\begin{equation}\label{eq:example_lm_tolerance}
\lmodel_1\left(a^n\right)\eqdef
\begin{cases}
\{a\mapsto 0.4, \terminal\mapsto 0.6\} & \text{if } n\in N_1\\
\{a\mapsto 0.6, \terminal\mapsto 0.4\} & \text{if } n\in N_2 
\end{cases}
\quad
\lmodel_2\left(a^n\right)\eqdef
\{a\mapsto 0.5, \terminal\mapsto 0.5\}\ \forall n\in\Nat,
\end{equation}
where $(N_1,N_2)$ is a partition of $\Nat$.
We consider the similarity relation $\tsim_{(\vd,\tol)}$ on $\Psimplex\left(\SymbT\right)$ with $\tol=0.15$.
Then for any choice of partition $(N_1, N_2)$ we have that $\lmodel_1\wsim_{(\vd,\tol)}\lmodel_2$.
But by choosing appropriately the partition $(N_1,N_2)$ we get different induced tolerance relations on $\Words$.
For instance, let $N_1 = \{n_k\}_k$ be a set with the property that the increments $n_{k+1}-n_k$ are strictly increasing.
Then the words $a^{n_k}$, with $k\geq 1$, are pairwise non related for $\wsim^{\lmodel_1}_{(\vd,\tol)}$.
On the other hand in $\wsim^{\lmodel_2}_{(\vd,\tol)}$ all words are related.
\qedexample
\end{Example}

~\\
Now, given a language model \lmodel\ and an equivalence $\Equ$, we define the relation $\indist_\Equ \subseteq \Words\times\Words$ as follows: 
\begin{align}\label{def:indist_str}
\woru \indist^\lmodel_{\Equ} \woru' &\iffdef 
\forall \worw\in\Words.\ \lmodel(\woru\worw) =_{\Equ} \lmodel(\woru'\worw) 
\end{align}
Indeed, $\indist^\lmodel_{\Equ}$ is a \emph{congruence}, that is, a transitive tolerance relation:\\

\begin{Proposition}\label{prop:congruence}
$\indist^\lmodel_{\Equ}$ is a \emph{congruence} on $(\Words, \Func)$.
\end{Proposition}
\begin{proof}
For reflexivity, symmetry and compatibility the proof is similar to Prop.~\ref{prop:tolerance_wsim}.
Transitivity follows from the transitivity of $\Equ$. \qedhere
\end{proof}

~\\
Equivalence $=_{\Equ}$ also induces an equivalence relation $\indist_{\Equ}$ between language models as follows:
\begin{align}\label{def:indist_lm}
\lmodel_1 \indist_{\Equ} \lmodel_2 
&\iffdef \forall\woru\in\Words.\ 
\lmodel_1(\woru) =_{\Equ} \lmodel_2(\woru) 
\end{align}
Indeed, $\lmodel_1$ and $\lmodel_2$ induce the same congruence over \Words. \\

\begin{Proposition}\label{prop:same_partition}
If $\lmodel_1 \indist_{\Equ} \lmodel_2$, then the congruences $\indist^{\lmodel_1}_{\Equ}$ and $\indist^{\lmodel_2}_{\Equ}$ are the same.
\end{Proposition}
\begin{proof}
Let us show that $\woru \indist^{\lmodel_1}_{\Equ} \woru'$ implies $\woru\indist^{\lmodel_2}_{\Equ}\woru'$, the other one follows by symmetry. Indeed, for all $\worw\in\Words$ we have
\begin{align*}
\lmodel_2(\woru\worw) & =_{\Equ} \lmodel_1(\woru\worw) & \text{by hypothesis and Def.~\ref{def:indist_lm}}\\
& =_{\Equ} \lmodel_1(\woru'\worw) & \text{by Def.~\ref{def:indist_str}}\\
& =_{\Equ} \lmodel_2(\woru'\worw) & \text{by hypothesis and Def.~\ref{def:indist_lm}}
\end{align*}
Then by transitivity $\lmodel_2(\woru\worw)=_{\Equ}\lmodel_2(\woru'\worw)$.
\end{proof}

~\\
Given the congruence $\indist^{\lmodel}_{\Equ}$, we write $\ccla{\Words}^\lmodel_{\Equ}$ for the partition of its equivalence classes and $\ccla{\woru}^\lmodel_{\Equ}$ for the class of $\woru\in\Words$. We use the dot $\ccla{\cdot}^\lmodel_\Equ$ for the quotient map $\woru\in\Words\mapsto\ccla{\woru}^\lmodel_\Equ\in\ccla{\Words}^\lmodel_\Equ$.
Notice that $\ccla{\Words}^\lmodel_\Equ$ is \emph{countable}.\\

\begin{Definition}\label{def:quotient_structure}
The congruence $\indist^\lmodel_{\Equ}$ defines a \emph{quotient} structure
\begin{align}\label{def:qlmodel}
\ccla{\lmodel}_\Equ &\eqdef \left(\ccla{\Words}^\lmodel_\Equ,\, \ccla{\emptyW}^\lmodel_\Equ,\, \qFunc,\, \qlmodel\right)
\end{align}
such that 
\begin{subequations}
\begin{align}
\qfunc_\symb\left(\ccla{\woru}^\lmodel_\Equ\right) 
    &\eqdef \ccla{\woru\symb}^\lmodel_{\Equ} \label{def:qfunc} \\
\qlmodel\left(\ccla{\woru}^\lmodel_\Equ\right)
    &\eqdef \ccla{\lmodel(\woru)}_\Equ \label{def:lmodeldot}
\end{align}    
\end{subequations}    
\end{Definition}
We define \qfuncW\ as the extension of \qfunc\ to strings as follows:
\begin{subequations}\label{def:qfuncW}
\begin{align}
\qfuncW(\emptyW)    &= \ccla{\emptyW}^\lmodel_\Equ \label{def:qfuncW_base} \\
\qfuncW(\woru\symb) &= \qfunc_\symb(\qfuncW(\woru)) \label{def:qfuncW_ih}
\end{align}    
\end{subequations}

\begin{Proposition}\label{prop:qfuncW}
For every $\woru\in\Words$, $\qfuncW(\woru) = \ccla{\woru}^\lmodel_\Equ$.
\end{Proposition}
\begin{proof}
By induction on the length of \woru.  
\begin{description}
    \item[Base case]
    By Def.~\ref{def:qfuncW_base}, $\qfuncW(\emptyW) = \ccla{\emptyW}^\lmodel_\Equ$.
    \item[Inductive step] 
    \begin{align*}
        \qfuncW(\woru\symb) 
            &= \qfunc_\symb(\qfuncW(\woru)) 
                &\text{by definition of \qfuncW}\\
            &= \qfunc_\symb(\ccla{\woru}^\lmodel_\Equ) 
                &\text{by I.H.}\\
            &= \ccla{\woru\symb}^\lmodel_\Equ
                &\text{by Def.~\ref{def:qfunc}} &\qedhere
    \end{align*}
\end{description}
\end{proof}

\begin{Corollary}\label{cor:algebraic_quotient}
$\lmodel_1 \indist_{\Equ} \lmodel_2$ if and only if
$\ccla{\lmodel_1}_\Equ = \ccla{\lmodel_2}_\Equ$.
\end{Corollary}
\begin{proof}
\begin{description}
\item[$\implies$] 
Suppose $\lmodel_1 \indist_{\Equ} \lmodel_2$.
Proposition~\ref{prop:same_partition} implies that
$\ccla{\Words}^{\lmodel_1}_{\Equ}$ = $\ccla{\Words}^{\lmodel_2}_{\Equ}$ 
and also that $\ccla{\emptyW}^{\lmodel_1}_{\Equ}$ = $\ccla{\emptyW}^{\lmodel_2}_{\Equ}$.
We show, $\qFunc_1 = \qFunc_2$:
\begin{align*}
\qfunc^1_\symb\left(\ccla{\woru}^{\lmodel_1}_\Equ\right) 
    &= \ccla{\woru\symb}^{\lmodel_1}_{\Equ}   
        &\text{by Def.~\ref{def:qfunc}} \\
    &= \ccla{\woru\symb}^{\lmodel_2}_{\Equ}   
        &\text{by Proposition~\ref{prop:same_partition}}
             \\
    &= \qfunc^2_\symb\left(\ccla{\woru}^{\lmodel_2}_\Equ\right) 
        &\text{by Def.~\ref{def:qfunc}} \\
    &= \qfunc^2_\symb\left(\ccla{\woru}^{\lmodel_1}_\Equ\right) 
        &\text{by Proposition~\ref{prop:same_partition}}
\end{align*}
$\qlmodel_1 = \qlmodel_2$ follows analogously:
\begin{align*}
\qlmodel_1\left(\ccla{\woru}^{\lmodel_1}_\Equ\right) 
    &= \ccla{\lmodel_1(\woru)}_{\Equ}   
        &\text{by Def.~\ref{def:lmodeldot}} \\
    &= \ccla{\lmodel_2(\woru)}_{\Equ}   
        &\text{by Def.~\ref{def:indist_lm}} \\
    &= \qlmodel_2\left(\ccla{\woru}^{\lmodel_2}_\Equ\right) 
        &\text{by Def.~\ref{def:lmodeldot}} \\
    &= \qlmodel_2\left(\ccla{\woru}^{\lmodel_1}_\Equ\right) 
        &\text{by Proposition~\ref{prop:same_partition}}
\end{align*}
Hence, $\ccla{\lmodel_1}_\Equ = \ccla{\lmodel_2}_\Equ$.
\item[$\impliedby$] 
Suppose $\ccla{\lmodel_1}_\Equ = \ccla{\lmodel_2}_\Equ$. 
Then, for all $\woru\in\Words$:
\begin{align*}
    \ccla{\lmodel_1\left(\woru\right)}_\Equ
        &= \qlmodel_1\left(\ccla{\woru}^{\lmodel_1}_\Equ\right) 
            &\text{by Def.~\ref{def:lmodeldot}} \\
        &= \qlmodel_2\left(\ccla{\woru}^{\lmodel_2}_\Equ\right) 
            &\text{by hypothesis} \\
        &= \ccla{\lmodel_2(\woru)}_\Equ
            &\text{by Def.~\ref{def:lmodeldot}}
\end{align*}
Hence, $\lmodel_1 \indist_{\Equ} \lmodel_2$.\qedhere
\end{description}
\end{proof}

~\\
We end this section by defining the concept of \Equ-\emph{regularity}.\\

\begin{Definition}\label{def:regular}
Given an equivalence \Equ, a language model \lmodel\ is \Equ-regular if $\ccla{\lmodel}_\Equ$ is finite.
\end{Definition}

\section{Probabilistic Deterministic Finite Automata}\label{sec:pdfa}

A \emph{probabilistic deterministic finite automaton} (PDFA)~\cite{IEEE:Vidal2005} over \Symb, denoted \Aut, is a tuple $(\Sta, \staI, \policy, \tra)$, where:
\begin{itemize}
    \item \Sta\ is a finite set of states,
    \item $\staI \in \Sta$ is an initial state, 
    \item $\policy : \Sta \ra \ProbT$ maps each state to a probability distribution over \SymbT, and 
    \item $\tra : \Sta \times \Symb \ra \Sta$ is the transition function. 
\end{itemize}
Both \policy\ and \tra\ are total functions.
We define $\traW$ to be the extension of \tra\ to \Words:
\begin{subequations}\label{def:traW}
\begin{align}
& \traW(\sta, \emptyW) \eqdef \sta 
\label{def:tau_star_a} \\
& \traW(\sta, \symb \woru)  \eqdef \traW( \tra(\sta, \symb), \woru ) \label{def:tau_star_b}
\end{align}
\end{subequations}
and $\policyW$ to be the extension of \policy\ to \Words:
\begin{align}\label{def:policyW}
\policyW(\sta, \woru) &\eqdef \policy( \traW( \sta, \woru) )       
\end{align}
When the state is \staI, we simply write $\traW(\woru)$ and $\policyW(\woru)$.
Without loss of generality, we assume that every state $\sta\in\Sta$ is \emph{reachable}, that is, $\sta = \traW(\woru)$ for some string $\woru\in\Words$. Any such \woru\ is called an \emph{access} string of \sta.

\subsection{Congruences defined by PDFA}

A PDFA \Aut\ defines the language model such that:
\begin{align}\label{def:Aut}
\lmodel_\Aut(\woru) &\eqdef \policyW(\woru)   
\end{align}
Now, Definition~\ref{def:simil_str} can be rephrased over \Sta:
\begin{align}\label{def:simil_sta}
\sta \wsim^\Aut_{\Tol} \sta' 
&\iffdef \forall \worw\in\Words.\ 
\policyW(\sta,\worw) \tsim_{\Tol} \policyW(\sta',\worw)
\end{align}
and similarly for Definition~\ref{def:indist_str}:
\begin{align}\label{def:indist_sta}
\sta \indist^\Aut_{\Equ} \sta' 
&\iffdef \forall \worw\in\Words.\ 
\policyW(\sta,\worw) =_{\Equ} \policyW(\sta',\worw)
\end{align}
This implies the following relationship between states and strings:\\

\begin{Proposition}\label{prop:t_suff_str_traW}
$\forall \woru, \woru' \in \Words.\ 
\woru \wsim^{\lmodel_\Aut}_{\Tol} \woru' \iff \traW(\woru) \wsim^\Aut_{\Tol} \traW(\woru')$.
\end{Proposition}
\begin{proof}
Let $\woru, \woru' \in \Words$:
\begin{align*}
    \woru \wsim^{\lmodel_\Aut}_{\Tol} \woru' 
    &\iff \forall \worw\in\Words.\ 
    \lmodel_\Aut(\woru\worw) \tsim_{\Tol} \lmodel_\Aut(\woru'\worw)                           & \text{by Def.~\ref{def:simil_str}}      \\
    &\iff \forall \worw\in\Words.\ 
    \policyW(\woru\worw) \tsim_{\Tol} \policyW(\woru'\worw)                   & \text{by Def.~\ref{def:Aut}}  \\
    &\iff \forall \worw\in\Words.\ 
    \policy(\traW(\woru\worw)) \tsim_{\Tol} \policy(\traW(\woru'\worw))      & \text{by Def.~\ref{def:policyW}}              \\
    &\iff \forall \worw\in\Words.\ 
    \policy(\traW(\traW(\woru), \worw)) \tsim_{\Tol} \policy(\traW(\traW(\woru'), \worw)) & \text{by Def.~\ref{def:traW}}     \\
    &\iff \forall \worw\in\Words.\ 
    \policyW(\traW(\woru), \worw) \tsim_{\Tol} \policyW(\traW(\woru'), \worw)             & \text{by Def.~\ref{def:policyW}}  \\
    &\iff \traW(\woru) \wsim^\Aut_{\Tol} \traW(\woru')                             & \text{by Def.~\ref{def:simil_sta}}       
    &\qedhere
\end{align*}
\end{proof}
\begin{Proposition}\label{prop:s_suff_str_traW}
$\forall \woru, \woru' \in \Words.\ 
\woru \indist^{\lmodel_\Aut}_{\Equ} \woru' \iff \traW(\woru) \indist^\Aut_{\Equ} \traW(\woru')$.
\end{Proposition}
\begin{proof}
Analogous to Proposition~\ref{prop:t_suff_str_traW}. 
\qedhere
\end{proof}
\begin{Proposition}\label{prop:congruence_indist}
1) $\wsim^\Aut_{\Tol}$ is a tolerance over $(\Sta,\tra)$.  
2) $\indist^\Aut_{\Equ}$ is a congruence over $(\Sta,\tra)$.
\end{Proposition}
\begin{proof}
Follows from Propositions~\ref{prop:tolerance_wsim}, \ref{prop:congruence}, \ref{prop:t_suff_str_traW}, and \ref{prop:s_suff_str_traW}.
\qedhere
\end{proof}

\subsection{Quotient PDFA}

Given an equivalence \Equ, we define a \emph{quotient PDFA} over $\Words$ as a tuple
\begin{align}\label{def:qpdfa}
\qpdfa &\eqdef \left(\bSta,\bstaI,\btra,\bpolicy\right)   
\end{align} 
where as in the case of PDFAs, $\bSta$ is a finite set of states, $\bstaI\in\bSta$ is an initial state, $\btra:\bSta\times\Symb\to\bSta$ is a transition function, and with the sole difference that the map $\bpolicy:\bSta\to\ccla{\ProbT}_\Equ$ associates an $\Equ$-equivalence class of probability distributions over $\SymbT$.
The extensions $\btraW$ and $\bpolicyW$ to $\Words$ are defined in an analogous way as in Definitions \ref{def:traW} and \ref{def:policyW}.

A PDFA $\Aut = \left(\Sta,\staI,\tra,\policy\right)$ is a \emph{realization} of the quotient PDFA \qpdfa\ if
\begin{align}\label{def:realization}
\Sta = \qSta, &\qquad 
\staI = \qstaI, \qquad 
\tra = \qtra, \qquad 
\forall\sta\in\Sta.\ 
\ccla{\policy\left(\sta\right)}_\Equ = \qpolicy\left(\sta\right)
\end{align}
%

Conversely, given a PDFA \Aut\ we can define its quotient PDFA $\ccla{\Aut}_\Equ$ as follows.
We denote $\qSta$ the set of equivalence classes $\ccla{\Sta}^\Aut_\Equ$ for the congruence defined in Def.~\ref{def:simil_sta}, and $q\in\Sta\mapsto\ccla{q}^\Aut_\Equ\in \qSta$ the associated quotient map. The transition function $\qtra$ is such that for all $\sta\in\Sta$ and $\symb\in\Symb$:
\begin{align}\label{def:quotient_Q}
\qtra\left(\ccla{\sta}^\Aut_\Equ, \symb\right) \eqdef \ccla{\tra(\sta,\symb)}^\Aut_\Equ 
\end{align}
which is well defined by Proposition~\ref{prop:congruence_indist}-2).
From Definition~\ref{def:indist_sta}, the composition $\ccla{\policy(\cdot)}_\Equ$ showed on the diagram in Fig.~\ref{fig:qpolicy}~(left) is constant on the equivalence classes $\qsta\in\qSta$, and therefore it factors through the quotient $\qSta$ giving the commutative diagram on the right.
\begin{figure}[h]
\centering
\begin{tikzcd}
\Sta \arrow[r, "\policy"] \arrow[dr, "\ccla{\policy(\cdot)}_\Equ"'] & \ProbT \arrow[d, "\ccla{\cdot}_\Equ"] \\
& \ccla{\ProbT}_\Equ
\end{tikzcd}
\qquad
\begin{tikzcd}
\Sta \arrow[r, "\policy"] \arrow[d, "\ccla{\cdot}^\Aut_\Equ"'] \arrow[dr, dotted] & \ProbT \arrow[d, "\ccla{\cdot}_\Equ"] \\
\qSta \arrow[r, "\qpolicy"] & \ccla{\ProbT}_\Equ
\end{tikzcd}
\caption{Definition of $\qpolicy$.\label{fig:qpolicy}}
\end{figure}

Summarizing, the quotient PDFA of \Aut\ is then:
\begin{align}\label{def:quotient_pdfa}
\ccla{\Aut}_\Equ &\eqdef ( \qSta, \qstaI, \qpolicy, \qtra )
\end{align}
where:
\begin{itemize}
    \item $\qSta\eqdef\ccla{\Sta}^\Aut_\Equ$, 
    \item $\qtra$ is given by Definition~\ref{def:quotient_Q}, 
    \item $\qstaI \eqdef \ccla{\staI}^\Aut_\Equ$, and
    \item $\qpolicy:\qSta\to \ccla{\ProbT}_\Equ$ is uniquely defined by $\qpolicy\left(\ccla{\sta}^\Aut_\Equ\right)=\ccla{\policy(\sta)}_\Equ$.
\end{itemize}

\begin{Proposition}\label{prop:qbijection}
The transition function $\traW:\Words \to \Sta$ induces an isomorphism between the quotients $\ccla{\lmodel_\Aut}_\Equ$ and $\ccla{\Aut}_\Equ$.
\end{Proposition}
\begin{proof}
Consider the diagram showed on Fig.~\ref{fig:strtosta}~(left). 
\begin{figure}[h]
\centering
\begin{tikzcd}
\Words \arrow[r, "\traW"] \arrow[dr, "\ccla{\traW(\cdot)}^\Aut_\Equ"'] & \Sta \arrow[d, "\ccla{\cdot}^\Aut_\Equ"] \\
& \qSta
\end{tikzcd}
\qquad
\begin{tikzcd}
\Words \arrow[r, "\traW"] \arrow[d, "\ccla{\cdot}^{\lmodel_\Aut}_\Equ"'] \arrow[dr, dotted] & \Sta \arrow[d, "\ccla{\cdot}^\Aut_\Equ"] \\
\ccla{\Words}^{\lmodel_\Aut}_\Equ \arrow[r, "\beta"] & \qSta
\end{tikzcd}
\caption{Definition of $\strtosta$.\label{fig:strtosta}}
\end{figure}
By the ``only if'' implication of Proposition~\ref{prop:s_suff_str_traW}, the composition $\ccla{\traW(\cdot)}^\Aut_\Equ$ is constant on the equivalence classes of $\indist^\Aut_{\Equ}$. Therefore this composition factors through the quotient $\ccla{\Words}^{\lmodel_\Aut}_\Equ$ giving a map \strtosta\ showed on the above diagram on the right. By construction \strtosta\ maps the class of a string to the class of a state:
\begin{align}\label{def:strtosta}
\strtosta\left(\ccla{u}^{\lmodel_\Aut}_\Equ\right) &\eqdef \ccla{\traW(\woru)}^\Aut_\Equ     
\end{align}
By the ``if'' implication of Proposition~\ref{prop:s_suff_str_traW}, the map \strtosta\ is injective. Since every state in $\Sta$ is reachable, \strtosta\ is surjective.

Let us show that \strtosta\ preserves the algebraic operations. For all $\woru\in\Words$ and $\symb\in\Symb$ we have:
\begin{align*}
\qtra\left(\strtosta\left(\ccla{\woru}^{\lmodel_\Aut}_\Equ\right), \symb\right) & = \qtra\left(\ccla{\traW(\woru)}^\Aut_\Equ,\symb\right) & \text{by Def.~\ref{def:strtosta} }\\
& = \ccla{\tra\left(\traW(\woru),\symb\right)}^\Aut_\Equ & \text{by Def.~\ref{def:quotient_Q}}\\
& = \ccla{\traW(\woru\symb)}^\Aut_\Equ & \text{by Def.~\ref{def:tau_star_b}}\\
& = \strtosta\left(\ccla{\woru\symb}^{\lmodel_\Aut}_\Equ\right)& \text{by Def.~\ref{def:strtosta}}\\
& = \strtosta\left(\qfunc_\symb\left(\ccla{\woru}^{\lmodel_\Aut}_\Equ\right)\right) & \text{by Def.~\ref{def:qfunc}}
\end{align*}
Also $\strtosta$ maps the class of the empty word to the initial state:
\begin{align*}
\strtosta\left(\ccla{\emptyW}^{\lmodel_\Aut}_\Equ\right) & = \ccla{\traW(\emptyW)}^\Aut_\Equ & \text{by Def.~\ref{def:strtosta}}\\
& = \ccla{\staI}^\Aut_\Equ & \text{by Def.~\ref{def:tau_star_a}}\\
& = \qstaI & \text{by Def.~\ref{def:quotient_pdfa}}
\end{align*}
Finally, 
\begin{align*}
\qpolicy\left(\strtosta\left(\ccla{\woru}^{\lmodel_\Aut}_\Equ\right)\right) & = \qpolicy\left(\ccla{\traW(\woru)}^\Aut_\Equ\right) & \text{by Def.~\ref{def:strtosta}}\\
& = \ccla{\policy\left(\traW(\woru)\right)}_\Equ & \text{by Fig.~\ref{fig:qpolicy}~(right)}\\
& = \ccla{\policyW(u)}_\Equ & \text{by Def.~\ref{def:policyW}}\\
& = \ccla{\lmodel_\Aut(u)}_\Equ & \text{by Def.~\ref{def:Aut}}\\
& = \qlmodel_\Aut\left(\ccla{u}^{\lmodel_\Aut}_\Equ\right) & \text{by Def.~\ref{def:lmodeldot}}
\end{align*}
This shows that \strtosta\ is an isomorphism.
\end{proof}

\begin{Corollary}\label{cor:qtraW}
For all $\woru\in\Words$ we have $\qtraW(\woru)=\strtosta\left(\qfuncW(\woru)\right)$.
\end{Corollary}
\begin{proof}
By induction on $\woru\in\Words$: 
\begin{description}
    \item[Base case]
    $\qtraW(\emptyW) = \qstaI = \strtosta\left(\ccla{\emptyW}^{\lmodel_\Aut}_\Equ\right)=\strtosta\left(\qfuncW(\emptyW)\right)$.
    \item[Inductive step] 
    \begin{align*}
        \qtraW(\woru\symb) 
            &= \qtra\left(\qtraW(\woru), \symb\right) 
                &\text{by definition of $\qtraW$}\\
            &= \qtra\left(\strtosta\left(\qfuncW(\woru)\right),\symb\right) & \text{by I.H.}\\
            &= \strtosta\left(\qfunc_\symb\left(\qfuncW(\woru)\right)\right)       
                &\text{by Prop.~\ref{prop:qbijection}}\\
            &= \strtosta\left(\qfuncW(\woru\symb)\right)
                &\text{by Def.~\ref{def:qfuncW_ih}}&\qedhere
    \end{align*}
\end{description}
\end{proof}


\begin{Corollary}\label{cor:quotient_cardinality}
Let $\Aut=(\Sta_\Aut,\staI^{\Aut}, \policy_\Aut, \tra_\Aut)$ and $\AutB=(\Sta_\AutB,\staI^{\AutB}, \policy_\AutB, \tra_\AutB)$ be two PDFA such that
$\lmodel_\Aut \indist_\Equ \lmodel_\AutB$.
Then $\ccla{\Aut}_\Equ$ and $\ccla{\AutB}_\Equ$ are isomorphic. 
In particular $\card\qSta_\Aut = \card\qSta_\AutB$.
\end{Corollary}
\begin{proof}
By Proposition~\ref{prop:qbijection} we have that $\ccla{\Aut}_\Equ$ and $\ccla{\lmodel_\Aut}_\Equ$ are isomorphic. Analogously, the same holds for $\ccla{\AutB}_\Equ$ and $\ccla{\lmodel_\AutB}_\Equ$. The claim then follows since $\ccla{\lmodel_\Aut}_\Equ$ and $\ccla{\lmodel_\AutB}_\Equ$ are the same by Corollary~\ref{cor:algebraic_quotient}.
\end{proof}

\begin{Example}\label{ex:pdfa_quotient}
    Fig.~\ref{fig:pdfa_quotient} depicts two PDFA $\Aut$ and $\AutB$. Notice that $\Aut$ is like $\lmodel_1$ in Eq.~\ref{eq:example_lm_tolerance} with $N_1=\{1\}$. 
    %
%
    Now, if we take $\quantp = 3$, we have that 
    $\lmodel_\Aut \indist_\quantp \lmodel_\AutB$   
    since $\{0.4, 0.5, 0.6\} \subset \left[\frac{1}{3}, \frac{2}{3} \right)$. From Proposition~\ref{prop:same_partition}, congruences $\indist^{\Aut}_\quantp$ and $\indist^{\AutB}_\quantp$ are the same, 
    by Corollary~\ref{cor:algebraic_quotient}, 
    $\ccla{\lmodel_\Aut}_\quantp = \ccla{\lmodel_\AutB}_\quantp$,
    and Corollary~\ref{cor:quotient_cardinality} implies 
    $\ccla{\Aut}_\quantp = \ccla{\AutB}_\quantp$. 
    Moreover, \AutB\ is a realization of $\ccla{\AutB}_\quantp$.
    However, \Aut\ is not a realization of $\ccla{\Aut}_\quantp$.
    %
    \qedexample
\end{Example}

\begin{figure}[htbp]
\centering
    \begin{minipage}[c]{0.65\textwidth}%
    \begin{tikzpicture}
        \node[state, initial] (qlambda) {\stackanchor{$\sta^\Aut_0$}{0.4}};
        \node[state, right of = qlambda] (qa) {\stackanchor{$\sta^\Aut_1$}{0.6}};
        \node[state, right of = qa] (qaa) {\stackanchor{$\sta^\Aut_2$}{0.4}};
        \draw           
            (qlambda) edge[above] node{$a/0.6$} (qa)
            (qa) edge[above] node{$a/0.4$}(qaa)
            (qaa) edge[loop right] node{$a/0.6$}(qaa);
    \end{tikzpicture}
    \end{minipage}
    \begin{minipage}[c]{0.3\textwidth}%
    \begin{tikzpicture}
        \node[state, initial] (q0) {\stackanchor{$q^\AutB_0$}{0.5}};
        \draw   
            (q0) edge[loop right] node{$a/0.5$} (q0);
    \end{tikzpicture}
    \end{minipage}
    \caption{(Left) $\Aut$. (Right) $\AutB$.}
    \label{fig:pdfa_quotient}
\end{figure}

Example~\ref{ex:pdfa_quotient} shows that a PDFA may not be realization of its quotient. On the other hand, every realization of the quotient of a PDFA \Aut\ is equivalent to \Aut. To show this, we first prove the following useful result.\\

\begin{Proposition}\label{prop:realization_of_qpdfa}
Given an equivalence \Equ\ and quotient PDFA \qpdfa, for every realization \Aut\ of \qpdfa\ we have that $\ccla{\lmodel_\Aut(\woru)}_\Equ = \qpolicyW(\woru)$ for all $\woru\in\Words$. 
\end{Proposition}
\begin{proof}
Let \Aut\ be a  realization of \qpdfa. By Def.~\ref{def:realization}, $\tra=\qtra$ and $\ccla{\policy(\sta)}_\Equ = \qpolicy(\sta)$.
Then, for all $\woru\in\Words$:
\begin{align*}
\ccla{\lmodel_\Aut(\woru)}_\Equ 
    &= \ccla{\policyW(\woru)}_\Equ 
            &\text{by Def.~\ref{def:Aut}} \\
    &= \ccla{\policy(\traW(\woru))}_\Equ 
            &\text{by Def.~\ref{def:policyW}} \\
    &= \qpolicy(\qtraW(\woru)) 
            &\text{$\Aut$ is a realization of \qpdfa} \\
    &= \qpolicyW(\woru)
            &\text{by Def.~\ref{def:policyW}}&\qedhere
\end{align*}
\end{proof}

\begin{Proposition}\label{prop:realizatioin_equiv}
For all PDFA \Aut, 
every realization \AutB\ of $\ccla{\Aut}_\Equ$ is such that 
$\lmodel_\AutB \indist_\Equ \lmodel_\Aut$.
\end{Proposition}
\begin{proof}
Let $\Aut = (\Sta, \staI, \tra, \policy)$,
$\ccla{\Aut}_\Equ = (\qSta, \qstaI, \qtra, \qpolicy)$,
and $\AutB = (\qSta, \qstaI, \tra_\AutB, \policy_\AutB)$,
with \AutB\ a realization of $\ccla{\Aut}_\Equ$.
For every $\woru\in\Words$:
\begin{align*}
\ccla{\lmodel_\AutB(\woru)}_\Equ 
    &= \qpolicyW(\woru)
            &\text{by Prop.~\ref{prop:realization_of_qpdfa}} \\
    &= \qpolicy(\qtraW(\woru)) 
            &\text{by Def.~\ref{def:policyW}} \\
    &= \qpolicy\left(\strtosta\left(\qfuncW(u)\right)\right) 
            &\text{by Cor.~\ref{cor:qtraW}} \\
    &= \qpolicy\left(\strtosta\left(\ccla{\woru}^{\lmodel_\Aut}_\Equ\right)\right) 
            &\text{by Prop.~\ref{prop:qfuncW}} \\
    &= \ccla{\lmodel_\Aut(\woru)}_\Equ 
            &\text{by Prop.~\ref{prop:qbijection}}
\end{align*}
Hence, $\lmodel_\AutB \indist_\Equ \lmodel_\Aut$. \qedhere
\end{proof}

\subsection{Minimality}

The following result states that any realization of $\ccla{\Aut}_\Equ$ is minimal, with respect to the number of states, among all PDFA which are $\indist_{\Equ}$-equivalent to \Aut. \\

\begin{Proposition}\label{prop:congruence_stasim_minimal}
For all PDFA $\Aut$,
\[
\card\qSta=\min_{\lmodel_\AutB\indist_\Equ \lmodel_\Aut}\card\Sta_\AutB
\]
where the minimum is taken over all PDFA $\AutB$ which are $\indist_\Equ$-equivalent to $\Aut$.
\end{Proposition}
\begin{proof}
Let $\AutB=(\Sta_\AutB,\staI^{\AutB}, \policy_\AutB, \tra_\AutB)$ be an arbitrary PDFA $\indist_{\Equ}$-equivalent to \Aut. Then
\[
\card\Sta_\AutB\geq \card\qSta_\AutB = \card\qSta_\Aut
\]
where the last equality follows from Corollary~\ref{cor:quotient_cardinality}. That the minimum is attained follows from Proposition~\ref{prop:realizatioin_equiv} by letting $\AutB$ to be any realization of $\ccla{\Aut}_\Equ$.
\end{proof}

\begin{Example} PDFA $\AutB$ in Fig.~\ref{fig:pdfa_quotient} is a realization of the quotient $\ccla{\Aut}_\quantp$, for $\quantp=3$, and therefore, it is a minimal PDFA $\indist_\quantp$-equivalent to $\Aut$.    
\qedexample
\end{Example}

\section{Learning with equivalence relations}\label{sec:plstar}
We present \PLstar, an adaptation of \Lstar~\cite{Angluin:1987} for language models. 
Given an unknown \emph{target} language model \lmodel\ and an equivalence \Equ, the goal of \PLstar\ is to learn a quotient PDFA \qpdfa\ isomorphic to $\ccla{\lmodel}_\Equ$.
If \lmodel\ is \Equ-regular, \PLstar\ is guaranteed to terminate.
Hereinafter, \lmodel\ and \Equ\ are fixed.

\subsection{Queries}
The algorithm makes use of a so-called \emph{membership query} \MQ\ defined as follows:
\begin{align}\label{def:MQ}
    \MQ(\woru) &\eqdef \lmodel(\woru)
\end{align} 
together with an \emph{equivalence query} \EQ\ defined as follows:
\begin{align}\label{def:EQ}
    \EQ(\qpdfa, \Equ) &\eqdef
    \begin{cases}
        \textsc{True} & \textrm{if}\ \forall\woru\in\Words.\ \ccla{\lmodel(\woru)}_\Equ = \qpolicyW(\woru) \\
        \ce  & \textrm{such that}\ \ccla{\lmodel(\ce)}_\Equ \neq \qpolicyW(\ce)
    \end{cases}
\end{align}
where \ce\ is called a \emph{counterexample}.

\subsection{\Wor-equivalence}
For any set of strings $\Wor\subseteq\Words$, we define:
\begin{align}\label{def:s_suff_str}
\forall \woru, \woru' \in \Words.\ \woru =_\Equ^{\Wor} \woru' &\iff \forall \worw\in\Wor.\ \lmodel(\woru\worw) =_\Equ \lmodel(\woru'\worw)
\end{align}
It is straightforward to show that $=_\Equ^{\Wor}$ is an equivalence relation.
Notice that $=_\Equ^{\Words}$ is $\indist^\lmodel_\Equ$. 
We denote $\cla{\cdot}_\Equ^{\Wor}$ the classes defined by the equivalence $=_\Equ^{\Wor}$. 

Recall that given two relations $R_1$ and $R_2$ on any set $X$, $R_1$ is \emph{finer} than $R_2$ if only if for all $x, y \in X$, $x R_1 y$ implies $x R_2 y$. It is also said that $R_2$ is \emph{coarser} than $R_1$. \\
\begin{Proposition}\label{prop:equal_equ}
Let $\Wor_1, \Wor_2 \subseteq \Words$ such that $\Wor_1 \subseteq \Wor_2$.
Then $=_\Equ^{\Wor_2}$ is finer than $=_\Equ^{\Wor_1}$.
\begin{proof}
Let $\woru, \woru' \in \Words$:
\begin{align*}
\woru =_\Equ^{\Wor_2} \woru' 
&\implies \forall \worw\in\Wor_2.\ \lmodel(\woru\worw) =_\Equ \lmodel(\woru'\worw) 
    & \text{by Def.~\ref{def:s_suff_str}} \\
&\implies \forall \worw\in\Wor_1.\ \lmodel(\woru\worw) =_\Equ \lmodel(\woru'\worw) 
    & \text{by}\ \Wor_1 \subseteq \Wor_2 \\
&\implies \woru =_\Equ^{\Wor_1} \woru' 
    & \text{by Def.~\ref{def:s_suff_str}} 
    & \qedhere
\end{align*}
\end{proof}
\end{Proposition}
%
%
\begin{Corollary}\label{coro:finer}
For all $\Wor\subseteq\Words$, $\indist^\lmodel_\Equ$ is finer than $=_\Equ^\Wor$. Moreover, if the quotient $\ccla{\Words}^\lmodel_\Equ$ is finite and $\card\cla{\Words}^W_\Equ = \card\ccla{\Words}^\lmodel_\Equ$, then $=^\Wor_\Equ$ and $\indist^\lmodel_\Equ$ are the same.
\end{Corollary}
\begin{proof}
The first claim follows directly from Proposition \ref{prop:equal_equ}. The second follows since equality of cardinals implies that the map $\ccla{\woru}^\lmodel_\Equ\mapsto\cla{\woru}^\Wor_\Equ$ is a bijection.
\end{proof}

\subsection{Algorithm \PLstar}
\PLstar\ pseudocode (Algorithm~\ref{alg:pl_star}) is analogue to \Lstar.
It uses an \emph{observation table}
\[
OT: \Pref \times \Suff \to \Psimplex\left(\SymbT\right)
\]
for storing outcomes of \MQ, where $\Pref \subset \Symb^\ast$ is a finite prefix-closed set (stored in row indices) and $\Suff \subset \Symb^\ast$ is a finite suffixed-closed set (stored in column indices). 
Given $\woru\in\Words$, we denote $\prefixes(\woru)$ and $\suffixes(\woru)$ the set of prefixes and suffixes of \woru, including \woru\ and \emptyW.
$OT$ is defined as follows: 
\begin{align}
\label{def:OT}
\forall p\in\Pref, s\in\Suff.\ OT[p][s] \eqdef \MQ(p s) 
\end{align}
\Pref\ is divided into two parts: a prefix-closed set \RED\ which are the rows used to construct the \emph{states} of the quotient PDFA \qpdfa, and $\BLUE\eqdef\left(\RED\right)\Symb$ which are the rows representing continuations of \RED\ by every symbol $\symb\in\Symb$~\cite{De_la_higuera:2010}. 
The fact that \RED\ is prefix-closed implies \Pref\ is also prefix-closed. 

\PLstar\ expands $OT$ through the use of \MQ\ until it becomes \emph{closed} and \emph{consistent} (lines~\ref{line:begin_while} to~\ref{line:end_while}). Then, it constructs a hypothesis quotient PDFA (line~\ref{line:build_pdfa}) and calls \EQ\ to check if it is equivalent to the target language model (line~\ref{line:EQ}). If \EQ\ returns a counterexample \ce, $OT$ is updated (line~\ref{line:update}). These steps are repeated until \EQ\ answers \True, in which case \PLstar\ terminates and returns the last hypothesis \qpdfa\ (line~\ref{line:return}).\\

\begin{figure}
\centering
\begin{minipage}{.6\textwidth}
\begin{algorithm}[H]
    \SetKwInOut{Input}{Input}
    \SetKwInOut{Output}{Output}
    
    \Input{ An alphabet $\Symb$, a language model $\lmodel$, an equivalence $\Equ$}
    \Output{ Quotient PDFA \qpdfa}

    Initialize\;
    \Repeat{Answer = \True}{
        \While{ $OT$ is not closed or not consistent}
        {\label{line:begin_while}
            \If{$OT$ is not closed}
            {
               $OT$ $\leftarrow$ \Close($OT$, $\Equ$)\; \label{line:close}
            }
            \If{$OT$ is not consistent}
            {
                $OT$ $\leftarrow$ \Consistent($OT$, $\Symb$, $\Equ$)\; \label{line:consistent}
            }
        } \label{line:end_while}
        \qpdfa $\leftarrow$ \BuildQPDFA($OT$, $\Equ$)\; \label{line:build_pdfa}
        $Answer$ $\leftarrow$ \EQ($\qpdfa, \Equ$)\; \label{line:EQ}      
        \If{Answer $=$ \ce}
        { 
            $OT$ $\leftarrow$ \Update($OT$, \ce)\; \label{line:update} 
        }       
    }
   
    \Return \qpdfa\; \label{line:return}
    
    \caption{$\PLstar$ learning algorithm}
    \label{alg:pl_star}
\end{algorithm}    
\end{minipage}
\end{figure}

\begin{description}

\item[Closedness]
$OT$ is \emph{closed} if and only if
\begin{align}\label{def:closed}
\forall p \in \BLUE,\ \exists p' \in \RED\ &\text{ such that }\ p =_\Equ^{\Suff} p'    
\end{align}
Equivalently, $OT$ is closed if and only if $\cla{\BLUE}^\Suff_\Equ \subseteq \cla{\RED}^\Suff_\Equ$.
While $OT$ is not closed, 
\Close\ finds $p'\in \BLUE$ such that $p'\neq^{\Suff}_\Equ p$ for all $p\in \RED$, and updates $OT$ as follows:
\begin{equation}\label{def:Closed}
\begin{aligned}
\RED'  &\gets \RED\cup\{p'\}\\ 
\BLUE' &\gets \BLUE\setminus\{p'\} \cup \{p' \symb \mid \symb\in\Symb \} \\
\Suff' &\gets \Suff\\
OT[p' \symb][s] &\gets \MQ(p' \symb s),\ \textrm{for all } \symb\in\Symb,\ s \in\Suff \\
\end{aligned}
\end{equation}
Notice that \RED\ remains prefix-closed, and so \Pref, and that \Suff\ remains unchanged.\\

\item[Consistency] 
$OT$ is \emph{consistent} if and only if 
\begin{align}\label{def:consistency}
\forall p, p' \in \RED,\ &\text{ if }\ p =_\Equ^{\Suff} p'\ \text{ then }\ 
\forall \symb \in \Symb.\ p \symb=_\Equ^{\Suff} p' \symb    
\end{align}
While $OT$ is not consistent, \Consistent\ finds two rows $p, p' \in \RED$ such that $p =_\Equ^{\Suff} p'$ but $p \symb s \neq_\Equ p' \symb s$ for some $s \in \Suff$, adds $\symb s$ to \Suff, and updates $OT$ as follows:
\begin{equation}\label{def:Consistent}
\begin{aligned}
\RED'  &\gets \RED\\ 
\BLUE' &\gets \BLUE \\
\Suff' &\gets \Suff\cup\{\sigma s\} \\
OT[p][\symb s] &\gets \MQ(p \symb s),\ p \in\Pref
\end{aligned}
\end{equation}
Notice that \Suff\ remains suffixed-closed and \Pref\ remains unchanged.\\

\item[Quotient PDFA construction]
\BuildQPDFA\ returns a quotient PDFA \qpdfa\  
such that:
\begin{subequations}
\begin{align}
    \hSta & \eqdef \cla{\RED}^\Suff_\Equ,\text{ and } \label{def:hSta}\\
    \hstaI & \eqdef \cla{\emptyW}^\Suff_\Equ. \label{def:hstaI} 
\end{align}
\end{subequations}
Closedness and consistency ensure that we can define a transition function
\[
\htra: \cla{\RED}^\Suff_\Equ\times \Symb\to\cla{\RED}^\Suff_\Equ
\]
by letting
\begin{equation}\label{def:tau_hat}
\htra\left(\cla{p}^\Suff_\Equ, \symb\right) \eqdef \cla{p \symb}^\Suff_\Equ
\end{equation}
The map \hpolicy\ is defined as follows:
\begin{equation}\label{def:pi_hat}
\hpolicy(\cla{p}^\Suff_\Equ) \eqdef \ccla{OT[p][\emptyW]}_\Equ,\quad p\in\RED.
\end{equation}
It is well defined by consistency.

\item[Update]
When \EQ\ returns a counterexample \ce, \Update\ adds the set $\prefixes(\ce)$ to \RED, expands \BLUE\ with the missing continuations, and fills $OT$ with appropriate \MQ s:

\begin{equation}\label{def:Update}
\begin{aligned}
\RED'  &\gets \RED\cup\prefixes(\ce)\\ 
\BLUE' &\gets \BLUE\cup \{p\symb:p\in\prefixes(\ce),\ \symb\in\Symb\} \\
\Suff' &\gets \Suff \\
OT[p][s] &\gets \MQ(ps),\ \text{for all } p\in\prefixes(\ce),\ s \in\Suff \\
OT[p\symb][s] &\gets \MQ(p\symb s),\ \text{for all } p\in\prefixes(\ce),\ \symb\in\Symb,\ s \in\Suff
\end{aligned}
\end{equation}

\end{description}

\subsection{Properties of the quotient PDFA built from an $OT$}

The following lemmas state basic properties of the quotient PDFA built from a closed and consistent observation table $OT$ via the procedure $\BuildQPDFA$. We will use them in the next section in the proof of termination of Algorithm~\ref{alg:pl_star}.\footnote{Lemma~\ref{lemma:plstar-quant-A} and Lemma~\ref{lemma:plstar-quant-B} are analogous to Angluin's Lemmas for regular languages~\cite{Angluin:1987}.}
It is worth mentioning that Lemma~\ref{lemma:plstar-quant-A} and Lemma~\ref{lemma:plstar-quant-B} are adapted versions of Proposition~\ref{prop:qfuncW} and Proposition~\ref{prop:realization_of_qpdfa}, respectively, that hold for equivalence $=^\Suff_\Equ$, subject to closedness and consistency of $OT$. \\

\begin{Lemma}\label{lemma:plstar-quant-A}
Let $OT$ be closed and consistent and \qpdfa\ be the quotient PDFA built from $OT$.
Then for all $p \in \RED$, we have $\htraW(p) = \cla{p}^\Suff_\Equ$.
\end{Lemma}
\begin{proof}
By induction over $p \in \RED$.
\begin{description}
\item[Base case] $p = \emptyW$.
By construction, 
$\htraW(\emptyW) = \hstaI = \cla{\emptyW}^\Suff_\Equ$.
\item[Inductive step] $p = p' \symb$.
\begin{align*}
\htraW(p' \symb) 
 &= \htra\left(\htraW(p'), \symb\right)   
    &\qquad \text{by definition of \htraW} \\ 
 &= \htra\left(\cla{p'}^\Suff_\Equ, \symb\right)  
    &\qquad \text{by IH and \RED\ prefix-closed} \\ 
 &= \cla{p'\symb}^\Suff_\Equ 
 &\qquad \text{by Def.~\ref{def:tau_hat}}
 &\qedhere
\end{align*}
\end{description}
\end{proof}

\allowdisplaybreaks
\begin{Lemma}\label{lemma:plstar-quant-B}
Let $OT$ be closed and consistent and \qpdfa\ be the quotient PDFA built from $OT$.
Then for all $p \in \RED$ and $s \in \Suff$, we have $\hpolicyW(p s) = \ccla{\lmodel(p s)}_\Equ$.
\end{Lemma}
\begin{proof}
By induction in the length of $s \in \Suff$.
\begin{description}
\item[Base case] Let $|s|=0$, i.e., $s = \emptyW$.
\begin{align*}
   \hpolicyW(p \emptyW) &= \hpolicyW(p) 
                            &\qquad p \emptyW = p \\
                        &= \widehat{\policy}\left(\htraW(p)\right) 
                            &\qquad \text{by definition of $\hpolicyW$}\\
                        &= \widehat{\policy}\left(\cla{p}^\Suff_\Equ\right) 
                            &\qquad \text{by Lemma \ref{lemma:plstar-quant-A}}\\
                        &= \ccla{OT[p][\emptyW]}_\Equ
                            &\qquad  \text{by Def.~\ref{def:pi_hat}} \\
                        &= \ccla{\MQ(p \emptyW)}_\Equ 
                            &\qquad \text{by Def.~\ref{def:OT}}\\
                        &= \ccla{\lmodel(p\emptyW)}_\Equ
                            &\qquad \text{by Def.~\ref{def:MQ}}
\end{align*}
\item[Inductive step] Assume it holds for all $s' \in \Suff$, with $|s'| = n$. Let $s = \symb s'$, with $|s'| = n$.
 \begin{align*}
    \hpolicyW(p \symb s') 
    &=  \hpolicy\left(\htraW(p \symb s')\right) 
        &\qquad \text{by definition of $\hpolicyW$}\\
    &=  \hpolicy\left(\htraW\left(\htraW(p), \symb s'\right)\right) 
        &\qquad \text{by definition of $\htraW$}\\
    &=  \hpolicy\left(\htraW\left(\cla{p}^\Suff_\Equ, \symb s'\right)\right) 
        &\qquad \text{by Lemma~\ref{lemma:plstar-quant-A}}\\
    &=  \hpolicy\left(\htraW\left(\htra\left(\cla{p}^\Suff_\Equ, \symb\right), s'\right)\right)
        &\qquad \text{by definition of $\htraW$}\\ 
    &=  \hpolicy\left(\htraW\left(\cla{p \symb}^\Suff_\Equ, s'\right)\right) 
        &\qquad \text{by Def.~\ref{def:tau_hat}}\\
    &=  \hpolicy\left(\htraW\left(\cla{p'}^\Suff_\Equ, s'\right)\right) 
        &\qquad \text{with }p'\in\RED\text{ since }OT\text{ closed}\\
    &=  \hpolicy\left(\htraW\left(\htraW\left(p'\right), s'\right)\right) 
        &\qquad \text{by Lemma~\ref{lemma:plstar-quant-A}}\\
    &=  \hpolicy\left(\htraW\left(p' s'\right)\right) 
        &\qquad \text{by definition of $\htraW$}\\
    &=  \hpolicyW\left(p's'\right) 
        &\qquad \text{by definition of $\hpolicyW$}\\
    &=  \ccla{\lmodel\left(p's'\right)}_\Equ  
        &\qquad \text{by IH: } p'\in \RED \text{ and } \left|s'\right|=n\\
    &=  \ccla{\lmodel(p \symb s')}_\Equ 
        &\qquad \text{since } \cla{p\symb}^\Suff_\Equ = \cla{p'}^\Suff_\Equ \text{ and }s'\in\Suff &\qedhere
 \end{align*}
\end{description} 
\end{proof}

\subsection{Correctness and termination}
We start by proving that \PLstar\ is correct. For this we need the following lemma.\\

\begin{Lemma}\label{lem:size-quotient}
Let $OT$ be an observation table. Then, $\card\cla{\RED}^{\Suff}_\Equ\leq \card\ccla{\Words}^\lmodel_\Equ$. In particular, if $OT$ is closed and consistent and \qpdfa\ is the quotient PDFA built from $OT$, then $\card \hSta\leq \card\ccla{\Words}^\lmodel_\Equ$.
\end{Lemma}
\begin{proof}
We have
\begin{align*}
\card\ccla{\Words}^\lmodel_\Equ 
    &\geq \card\cla{\Words}^\Suff_\Equ  
        &\text{by Corollary~\ref{coro:finer}} \\
    &\geq \card\cla{\RED}^\Suff_\Equ
        &\text{by}\ \RED\subseteq\Words
\end{align*}
When $OT$ is closed and consistent, we have $\hSta=\cla{\RED}^\Suff_\Equ$ by Definition~\ref{def:hSta}, and therefore $\card\hSta\leq \card\ccla{\Words}^\lmodel_\Equ$.
\end{proof}

\begin{Proposition}\label{prop:plstar-correctness-realization}
For any equivalence \Equ, quotient PDFA \qpdfa, and language model \lmodel, if $\EQ(\qpdfa,\Equ)$ returns \textsc{True}, then for every realization \Aut\ of \qpdfa, $\ccla{\lmodel}_\Equ = \ccla{\lmodel_\Aut}_\Equ$.    
\end{Proposition}
\begin{proof}
Let \Aut\ be any realization of \qpdfa. 
Then, for all $\woru\in\Words$:
\begin{align*}
\ccla{\lmodel(\woru)}_\Equ 
    &= \qpolicyW(\woru)
        &\text{by Def.~\ref{def:EQ}} \\
    &= \ccla{\lmodel_\Aut(\woru)}_\Equ 
        &\text{by Prop.~\ref{prop:realization_of_qpdfa}}
\end{align*}
Def.~\ref{def:indist_lm} implies $\lmodel \indist_\Equ \lmodel_\Aut$.
Hence, by Corollary~\ref{cor:algebraic_quotient},
$\ccla{\lmodel}_\Equ = \ccla{\lmodel_\Aut}_\Equ$. \qedhere
\end{proof}

\begin{Proposition}\label{prop:plstar-correctness}
For any language model \lmodel, if \PLstar\ terminates, it computes a quotient PDFA isomorphic to $\ccla{\lmodel}_\Equ$.
\end{Proposition}
\begin{proof}
If \PLstar\ terminates with \qpdfa, it means $\EQ(\qpdfa,\Equ)$ returns \textsc{True}.
Then, Proposition~\ref{prop:plstar-correctness-realization} implies  $\ccla{\lmodel}_\Equ = \ccla{\lmodel_\Aut}_\Equ$.
By Proposition~\ref{prop:qbijection}, it follows that $\ccla{\lmodel_\Aut}_\Equ$ is isomorphic to $\ccla{\Aut}_\Equ$.
Then
\begin{align*}
\card\ccla{\Words}^\lmodel_\Equ & =\card\hSta_\Aut & \text{by the above}\\
& \leq \card\Sta_\Aut & \text{by Proposition~\ref{prop:congruence_stasim_minimal}}\\
& = \card\hSta & \text{by definition of realization}\\
& \leq \card\ccla{\Words}^\lmodel_\Equ & \text{by Lemma~\ref{lem:size-quotient}}
\end{align*}
Hence:
\begin{equation}\label{eq:equal_cardinal}
\card\hSta_\Aut = \card\Sta_\Aut=\card\hSta=\card\ccla{\Words}^\lmodel_\Equ
\end{equation} 
Therefore $\qpdfa$ is isomorphic to $\ccla{\Aut}_\Equ$.
Then, \qpdfa\ is isomorphic to $\ccla{\lmodel}_\Equ$.
\qedhere
\end{proof}

\begin{Corollary}\label{cor:minimal}
If $\qpdfa$ is the resulting quotient PDFA returned by \PLstar, then any realization of $\qpdfa$ is a minimal PDFA $\indist_\Equ$-equivalent to $\lmodel$.
\end{Corollary}
\begin{proof}
Let $\Aut$ be a realization of $\qpdfa$. Then, 
$\lmodel \indist_\Equ \lmodel_\Aut$ by Prop.~\ref{prop:plstar-correctness-realization}, and 
$\ccla{\Aut}_\Equ$ is isomorphic to \qpdfa\ and $\ccla{\lmodel}_\Equ$,
with
$\card\Sta_\Aut=\card\hSta$ by Eq.~\ref{eq:equal_cardinal} in the proof of Prop.~\ref{prop:plstar-correctness}.
Moreover, any PDFA $\AutB$ such that $\lmodel_\AutB \indist_\Equ\lmodel$ also satisfies $\lmodel_\AutB \indist_\Equ \lmodel_\Aut$. 
Then, by Prop.~\ref{prop:congruence_stasim_minimal},
$\card\Sta_\Aut\leq \card\Sta_\AutB$.
\end{proof}

~\\
To prove termination we need to show some auxiliary results.\\

\begin{Lemma}\label{lem:close_terminates}
Let $OT$ be a non-closed observation table and $OT'$ be the result of the procedure given in \ref{def:Closed}. Then
\[
\card\cla{\RED'}^{\Suff'}_\Equ > \card\cla{\RED}^\Suff_\Equ
\]
\end{Lemma}
\begin{proof}
Since $OT$ is not closed we have that $\cla{\BLUE}^\Suff_\Equ\not\subseteq\cla{\RED}^\Suff_\Equ$. Procedure \ref{def:Closed} finds $p'\in \BLUE\setminus\RED$ such that $p'\neq^{\Suff}_\Equ p$ for all $p\in \RED$.
Since $p'\in\RED'\setminus\RED$ and $\Suff'=\Suff$, we have $\card\cla{\RED'}^{\Suff'}_\Equ>\card\cla{\RED}^{\Suff}_\Equ$.
\end{proof}

\begin{Lemma}\label{lem:consistent_terminates}
Let $OT$ be a non-consistent observation table and $OT'$ be the result of procedure given in \ref{def:Consistent}. Then
\[
\card\cla{\RED'}^{\Suff'}_\Equ > \card\cla{\RED}^\Suff_\Equ
\]
\end{Lemma}
\begin{proof}
Procedure \ref{def:Consistent} finds $p, p' \in \RED$ such that $p =_\Equ^{\Suff} p'$ but $p \symb s \neq_\Equ p' \symb s$ for some $\symb\in\Symb$ and $s \in \Suff$.
Since $\Suff\subset \Suff'$, then by Proposition~\ref{prop:equal_equ} we have that $=^{\Suff'}_\Equ$ is finer than $=^{\Suff}_\Equ$ over all $\Words$. Since $\symb s\in\Suff'\setminus\Suff$, we have $\cla{p}^{\Suff'}_\Equ\neq\cla{p'}^{\Suff'}_\Equ$, which implies $=^{\Suff'}_\Equ$ is strictly finer than $=^{\Suff}_\Equ$ over $\RED$. Therefore $\card\cla{\RED'}^{\Suff'}_\Equ=\card\cla{\RED}^{\Suff'}_\Equ>\card\cla{\RED}^{\Suff}_\Equ$.
\end{proof}

\begin{Lemma}\label{lemma:plstar-quant-E}
Let $OT_i$ be closed and consistent, $\ce\in\Words$ a counterexample, and $OT_{i+1}$ the new closed and consistent table obtained by the algorithm in the next iteration. 
Then
\[
\card\cla{\RED_i}^{\Suff_i}_\Equ < \card\cla{\RED_{i+1}}^{\Suff_{i+1}}_\Equ.
\]
\end{Lemma}
\begin{proof}
Let $OT'_i$ be the resulting table after \Update\ (Procedure \ref{def:Update}) and \Close.
%
By Lemma~\ref{lemma:plstar-quant-B}, $\ce\not\in\RED_i$. Therefore, $\RED_i\subset \RED_i'$ strictly. Also $\Suff_i'=\Suff_i$. We have to consider the following two cases.
\begin{description}
\item[Case 1 -- $\cla{\RED_i}^{\Suff_i}_\Equ \subset \cla{\RED'_i}^{\Suff'_i}_\Equ$ strictly.]
In this case we have:
\begin{align*}
    \card\cla{\RED_i}^{\Suff_i}_\Equ 
        &< \card\cla{\RED'_i}^{\Suff_i}_\Equ
            & \text{by case hypothesis}\\
        &\leq \card\cla{\RED_{i+1}}^{\Suff_i}_\Equ
            & \text{by}\ \RED'_i \subseteq \RED_{i+1}   \\ 
        &\leq \card\cla{\RED_{i+1}}^{\Suff_{i+1}}_\Equ
            & \text{by Prop.~\ref{prop:equal_equ},}\ 
                \Suff_i \subseteq \Suff_{i+1}  
\end{align*}
Hence,
$\card\cla{\RED_i}^{\Suff_i}_\Equ < \card\cla{\RED_{i+1}}^{\Suff_{i+1}}_\Equ$. \\

\item[Case 2 -- $\cla{\RED'_i}^{\Suff'_i}_\Equ = \cla{\RED_i}^{\Suff_i}_\Equ$.]
Let us first show that $OT'_i$ is not consistent.
Suppose on the contrary that $OT'_i$ is consistent. 
Since $OT'_i$ is closed, we have that $OT_{i+1} = OT'_i$.
Then, $\hSta_{i+1} = \hSta_i$ and so $\hpolicyW_{i+1} = \hpolicyW_i$.
Since $\ce\in\RED_{i+1}$ and $\emptyW\in\Suff_{i+1}$, 
Lemma~\ref{lemma:plstar-quant-B} implies 
$\hpolicyW_{i+1}(\ce) = \ccla{\lmodel(\ce)}_\Equ$.
Since $\ce$ is a counterexample, we have that $\hpolicyW_i(\ce) \neq \ccla{\lmodel(\ce)}_\Equ$,
which is a contradiction.
    
Let $OT''_i=\Consistent(OT'_i,\Equ)$, thus
\[
\begin{aligned}
\card\cla{\RED_i}^{\Suff_i}_\Equ 
    &= \card\cla{\RED'_i}^{\Suff'_i}_\Equ 
        & \text{by case hypothesis}\\
    &< \card\cla{\RED''_i}^{\Suff''_i}_\Equ 
        & \text{by Lemma~\ref{lem:consistent_terminates}}\\
    &\leq \card\cla{\RED_{i+1}}^{\Suff_{i+1}}_\Equ 
        & \text{by Prop.~\ref{prop:equal_equ},}\ 
                \Suff''_i \subseteq \Suff_{i+1},
                \RED''_i \subseteq \RED_{i+1}
\end{aligned}
\]
\end{description} 
In both cases we conclude $\card\cla{\RED_i}^{\Suff_i}_\Equ < \card\cla{\RED_{i+1}}^{\Suff_{i+1}}_\Equ$.
\end{proof}

\begin{Corollary}\label{cor:whileloop}
If $\lmodel$ is $\Equ$-regular then the while loop from lines \ref{line:begin_while} -- \ref{line:end_while} terminates.
\end{Corollary}
\begin{proof}
By Lemmas \ref{lem:close_terminates} and \ref{lem:consistent_terminates} imply that $\card\cla{\RED}^\Suff_\Equ$ strictly increases. By Lemma \ref{lem:size-quotient}, we have that $\card\cla{\RED}^\Suff_\Equ\leq \card\ccla{\Words}^\lmodel_\Equ$ which is finite by  Definition~\ref{def:regular}. Therefore \Close, \Consistent, and the while loop from lines \ref{line:begin_while} -- \ref{line:end_while} terminate.
\end{proof}

\begin{Proposition}\label{prop:pLstar-termination}
If $\lmodel$ is $\Equ$-regular then \PLstar\ terminates.   
\end{Proposition}
\begin{proof}
Let $\qpdfa_i$ be the quotient PDFA obtained at the $i$-th iteration of the main loop of Algorithm~\ref{alg:pl_star}.
Suppose $\EQ(\qpdfa_i, \Equ)$ returns a counterexample which is used to update the observation table. 
In this case, Corollary~\ref{cor:whileloop} implies the while loop in lines \ref{line:begin_while}-\ref{line:end_while} terminates and results in a closed and consistent observation table $OT_{i+1}$. Then, by Lemma~\ref{lemma:plstar-quant-E} the sequence $\left\{\card\cla{\RED_i}^{\Suff_{i}}_\Equ\right\}_i$ is strictly increasing. By Lemma~\ref{lem:size-quotient} it is bounded by $\card\ccla{\Words}^\lmodel_\Equ$, and therefore it must be finite.
Hence, \PLstar\ terminates.
\end{proof}

\begin{Theorem}\label{them:plstar-correctness_and_termination}
If $\lmodel$ is $\Equ$-regular then \PLstar\ terminates and computes a quotient PDFA isomorphic to $\ccla{\lmodel}_\Equ$.
\end{Theorem}
\begin{proof}
By Proposition~\ref{prop:plstar-correctness} and Proposition~\ref{prop:pLstar-termination}.    
\end{proof}

\subsection{PDFA \Equ-recognizability}

\begin{Definition}\label{def:E_recongizable}
Given an equivalence \Equ, we say a language model \lmodel\ is PDFA \Equ-\emph{recognizable} if there exists a PDFA \Aut\ such that $\lmodel \indist_\Equ \lmodel_\Aut$.   
\end{Definition}

~\\
For any equivalence \Equ, \Equ-regularity and PDFA \Equ-recognizability coincide.\\

\begin{Theorem}\label{thm:E_regular_recognizable}
For every equivalence \Equ\ and language model \lmodel, \lmodel\ is \Equ-regular if and only if \lmodel\ is PDFA \Equ-recognizable.   
\end{Theorem}
\begin{proof}
\begin{description}
\item[$\implies$] 
Suppose \lmodel\ is \Equ-regular. 
Then, by Prop.~\ref{prop:pLstar-termination}, \PLstar\ terminates.
Let \qpdfa\ be the output and \Aut\ be any realization of it.
By Prop.~\ref{prop:plstar-correctness-realization}, 
$\lmodel \indist_\Equ \lmodel_\Aut$.
Therefore, \lmodel\ is PDFA \Equ-recognizable.
\item[$\impliedby$] 
Suppose \lmodel\ is PDFA \Equ-recognizable. Then, there exists a PDFA \Aut\ such that $\lmodel \indist_\Equ \lmodel_\Aut$. 
By Corollary~\ref{cor:algebraic_quotient}, 
$\ccla{\lmodel}_\Equ = \ccla{\lmodel_\Aut}_\Equ$ and by Proposition~\ref{prop:qbijection}, $\ccla{\lmodel_\Aut}_\Equ$ is finite.
Therefore, \lmodel\ is \Equ-regular.
\qedhere
\end{description}
\end{proof}

\section{Learning with tolerance relations}\label{sec:tolerance}

Other works proposed active learning algorithms based on tolerance relations~\cite{weiss_WFA_learning}.
However, relying on a tolerance rather than on an congruence has two important consequences:
\begin{enumerate}
    \item Given a language model $\lmodel$ and a similarity $\Tol$ on $\ProbT$, there is no well defined notion of quotient structure for the tolerance $\wsim^\lmodel_\Tol$ as the one given in Definition \ref{def:qlmodel}.
    However, the concept of quotient gives a clear objective for learning and its key for termination (Proposition~\ref{prop:pLstar-termination}) which relies on the minimality of the quotient (Proposition~\ref{prop:congruence_stasim_minimal}).
    \item Given two tolerant language models $\lmodel_1$ and $\lmodel_2$, as in Definition \ref{def:tolerance_lm}, the tolerance relations $\wsim^{\lmodel_1}_\Tol$ and $\wsim^{\lmodel_2}_\Tol$ are not necessarily the same (see Example~\ref{ex:diff_tolerances}).
    However, equality of the relations for congruences is used in the proof of Proposition~\ref{prop:plstar-correctness} (correctness) which relies on Proposition~\ref{prop:same_partition} through Corollary~\ref{cor:algebraic_quotient}.
\end{enumerate}
Therefore, since the existence of a quotient structure and the equality of congruences are  cornerstone for correctness and termination of Algorithm~\ref{alg:pl_star}, it is worth studying the impact on learning when these properties do not hold.

%
We start by proving that if we are given an equivalence \Equ\ finer than a similarity \Tol, then Algorithm~\ref{alg:pl_star} can be used to learn a PDFA \Tol-tolerant to a target language model. \\
%

\begin{Proposition}
Let $\Equ$ be an equivalence finer than similarity $\Tol$ on $\ProbT$. Then, for any language model $\lmodel$, if $\Aut$ is a realization of the quotient PDFA output of \PLstar\ for \lmodel, then $\lmodel_\Aut\wsim_{\Tol}\lmodel$.    
\end{Proposition}
\begin{proof}
By Proposition~\ref{prop:plstar-correctness} we have $\lmodel_\Aut\indist_\Equ\lmodel$. By definition this means $\lmodel_\Aut(\woru)=_\Equ\lmodel(\woru)$ for any $\woru\in\Words$. Thus, for any $\woru\in\Words$, we have $\lmodel_\Aut(\woru)\tsim_\Tol\lmodel(\woru)$ since $\Equ$ is finer than $\Tol$. That is $\lmodel_\Aut\wsim_\Tol\lmodel$.
\end{proof}

\begin{figure}[htbp]
\centering
    \begin{tikzpicture}
        \node[state, initial] (qlambda) {\stackanchor{$\sta_0$}{0.6}};
        \node[state, right of = qlambda] (qa) {\stackanchor{$\sta_1$}{0.5}};
        \node[state, right of = qa] (qaa) {\stackanchor{$\sta_2$}{0.4}};
        \draw           
            (qlambda) edge[above] node{$a/0.4$} (qa)
            (qa) edge[above] node{$a/0.5$}(qaa)
            (qaa) edge[loop right] node{$a/0.6$}(qaa);
    \end{tikzpicture}
    \caption{PDFA $\Aut$}
    \label{fig:pdfa_quant_finer_tol}
\end{figure}

\begin{Example}\label{ex:quant_finer_tol}
By Proposition~\ref{prop:quantp_vs_tol} the equivalence $=_\quantp$ is finer than similarity $\tsim_{(\vd,\quantp^{-1})}$. Then $\mathrm{L}^\ast_{=_\quantp}$ can be used to learn a $(\vd,\quantp^{-1})$-tolerant PDFA.
Moreover, $\tsim_{(\vd,\tol_1)}$ is finer than $\tsim_{(\vd,\tol_2)}$ whenever $\tol_1\leq\tol_2$.
Then, for $\tol=0.15$, taking the quantization equivalence $=_\quantp$ with $\quantp\geq 7$, $\mathrm{L}^\ast_{=_\quantp}$ will return a quotient PDFA such that every realization is $(\vd,\quantp^{-1})$-tolerant with the target model, and therefore, $(\vd,0.15)$-tolerant.
For the PDFA \Aut\ in Figure~\ref{fig:pdfa_quant_finer_tol}, $\mathrm{L}^\ast_{=_\quantp}$ will return the quotient PDFA $\ccla{\Aut}_{=_\quantp}$ which has three states, since $\sta_0$, $\sta_1$, and $\sta_2$ are not $\indist_{=_\quantp}$-equivalent.
\qedexample
\end{Example}

~\\
Now, recall that, given a reflexive and symmetric relation $R$ in any set $X$, a \emph{clique} $\clique$ is a set of pairwise related elements in $X$ such that $x R y$ for all $x, y \in \clique$. A \emph{clique partition} $\Clique\subseteq\powerset{X}$ is a \emph{cover} of $R$ with pairwise disjoint cliques. 
Notice that a clique partition defines an equivalence relation $E$ that is finer than $R$ by letting $xEy$ if and only if $x$ and $y$ belong to the same clique of $C$.
Conversely, given an equivalence relation $E$ finer than $R$, the set of classes of $E$ defines a clique partition of $X$.
For $x\in X$, we denote $\ccla{x}_\Clique$ the clique of (the clique partition) $\Clique$ containing $x$. \\

\begin{Example}\label{ex:clique_tol}
Consider again the PDFA \Aut\ in Figure~\ref{fig:pdfa_quant_finer_tol}.
There are three clique partitions of the set of distributions of \Aut\ induced by $\tsim_{(\vd,0.15)}$, namely:
\begin{align*}
\Clique_1 &= \{\ \{[0.5, 0.5]\},\ \{[0.4, 0.6]\},\ \{[0.6, 0.4]\}\ \} \\
\Clique_2 &= \{\ \{[0.5, 0.5], [0.4, 0.6]\},\ \{[0.6, 0.4]\}\ \} \\
\Clique_3 &= \{\ \{[0.4, 0.6]\},\ \{[0.5, 0.5], [0.6, 0.4]\}\ \}    
\end{align*}
where $[x_1,x_2]$ is a shorthand for $\{a\mapsto x_1, \terminal\mapsto x_2\}$.
Clearly, $\Clique_1$ gives the same classes than $=_7$, so the output of $\mathrm{L}^\ast_{\Clique_1}$ is the same as $\mathrm{L}^\ast_{=_\quantp}$.
For $\Clique_2$, the output quotient PDFA still has three states because 
\begin{align*}
\ccla{\policy(\tra(\sta^\Aut_0, a))}_{\Clique_2} 
&= \{[0.5, 0.5], [0.4, 0.6]\} 
\neq \{[0.6, 0.4]\} 
= \ccla{\policy(\tra(\sta^\Aut_1, a))}_{\Clique_2} 
\end{align*}
which implies that $\sta^\Aut_0 \not\indist^\Aut_{\Clique_2} \sta^\Aut_1$.
On the other hance, for $\Clique_3$, the output quotient PDFA has two states because
\begin{align*}
\ccla{\policy(\tra(\sta^\Aut_1, a))}_{\Clique_3} 
&= \{[0.5, 0.5], [0.6, 0.4]\} 
= \ccla{\policy(\tra(\sta^\Aut_2, a))}_{\Clique_3} 
\end{align*}
which implies that $\sta^\Aut_1 \indist^\Aut_{\Clique_3} \sta^\Aut_2$.
\qedexample
\end{Example}

~\\
Another natural alternative consists in defining a clique partition of \Words. This idea is found in~\cite{weiss_WFA_learning}. The following definition generalizes this concept. 
Given a similarity \Tol\ and a language model \lmodel, a $(\Tol,\lmodel)$-\emph{clique congruence} is a clique partition of \Words\ induced by tolerance $\wsim^\lmodel_\Tol$ that satisfies:
\begin{align}\label{def:clique_congruence}
\forall\woru,\worv\in\Words.\ 
    \forall\symb\in\Symb.\ 
        \ccla{\woru}_\Clique = \ccla{\worv}_\Clique\
                \implies \ccla{\woru\symb}_\Clique = \ccla{\worv\symb}_\Clique
\end{align}

\subsection{\Tol-regularity and \Tol-recognizability}

\begin{Definition}\label{def:S_regular}
Given a similarity $\Tol$ we say that a language $\lmodel$ is $\Tol$-\emph{regular} if there exists a finite $(\Tol,\lmodel)$-clique congruence.
\end{Definition}

~\\

\begin{Definition}\label{def:S_recognizable}
Given a similarity $\Tol$ we say that a language $\lmodel$ is PDFA $\Tol$-\emph{recognizable} if there exists a PDFA \Aut\ such that $\lmodel \wsim_\Tol \lmodel_\Aut$.
\end{Definition}

~\\
\Tol-regularity implies PDFA \Tol-recognizability. \\

\begin{Proposition}\label{prop:S_regular_implies_recongnizable}
For every similarity \Tol\ and a language model \lmodel, if \lmodel\ is \Tol-regular then it is \Tol-recognizable by a PDFA. 
\end{Proposition}
\begin{proof}
Let $\Clique$ be a finite $(\Tol,\lmodel)$-clique congruence. We can build the following PDFA $\Aut_\Clique = ( \Sta, \staI, \policy, \tra )$ where:
\begin{itemize}
    \item$\Sta\eqdef\Clique$
    \item$\staI \eqdef \ccla{\emptyW}_\Clique$
    \item$\tra\left(\ccla{\woru}_\Clique, \symb\right) \eqdef \ccla{\woru \symb}_\Clique$
    \item$\policy\left(\ccla{\woru}_\Clique\right) \eqdef \text{ choose an arbitrary element of }\{\lmodel(\worv) \mid \worv\in\ccla{\woru}_\Clique\}$
\end{itemize}
By Definition~\ref{def:clique_congruence}\ $\tra$ is well defined, and by an argument analogous to Proposition~\ref{prop:qfuncW}, it satisfies $\traW(\woru)=\ccla{\woru}_\Clique$ for all $\woru\in\Words$. Then, for all $\woru\in\Words$, we have
\begin{align*}
\lmodel_{\Aut_\Clique}(\woru) &  = \policy\left(\ccla{\woru}_\Clique\right)& \text{since $\traW(\woru)=\ccla{\woru}_\Clique$}\\
& = \lmodel(\worv) & \text{for some $\worv\in\ccla{\woru}_\Clique$}\\
& \tsim_\Tol \lmodel(\woru) & \text{by definition of clique and $\worv\wsim^\lmodel_\Tol \woru$}
\end{align*}
Hence $\lmodel_{\Aut_\Clique}\wsim_\Tol\lmodel$.
\end{proof}

~\\
On the contrary, PDFA \Tol-recognizability does not imply \Tol-regularity. \\
\begin{Proposition}\label{prop:S_regular_vs_recongnizable}
There exists a similarity \Tol\ and a PDFA \Tol-recognizable language model \lmodel\ which is not \Tol-regular.   
\end{Proposition}
\begin{proof}
Let $\Tol$ be $\tsim_{(\vd,\tol)}$ for $\tol=0.15$, \lmodel\ be the language model $\lmodel_1$ of Example~\ref{ex:diff_tolerances} where $N_1 = \{n_k\}_k$ is a set with the property that increments $n_{k+1}-n_k$ are strictly increasing, and \AutB\ the PDFA of Figure~\ref{fig:pdfa_quant_finer_tol} (right).
Then $\lmodel \wsim_{\Tol} \lmodel_{\AutB}$.
Suppose by contradiction that $\lmodel$ is $\Tol$-regular.
Then there exists a finite $(\Tol,\lmodel)$-clique congruence $\Clique$ of $\Words$.
Since the set $\{a^{n_k}\}_{k\geq 1}$ is infinite, there exists a clique $\clique\in\Clique$ with $\{a^{n_i},a^{n_j}\}\subset \clique$ and $i\neq j$.
This is a contradiction since all words in $\{a^{n_k}\}_{k\geq 1}$ are pairwise non related by $\wsim^\lmodel_\Tol$.
Hence, $\lmodel$ is not $\Tol$-regular.
\qedhere
\end{proof}
This example also shows that there are language models \lmodel\ and tolerances \Tol\ such that no learning algorithm based on constructing the $(\Tol,\lmodel)$-clique congruence could learn an \Tol-tolerant PDFA of \lmodel\ even if such PDFA exists. 

\section{Conclusions}\label{sec:conclusions}

The paper studied the problem of learning probabilistic deterministic finite automata from language models through the lens of tools given by the algebraic structures induced by similarities and equivalences, which provides a framework for understanding the foundations of algorithms proposed and implemented in the literature.
On one hand, it shows that relying on equivalences on distributions allows solving the problem using the same artifacts than for formal languages, which are derived from the fact that there is a canonically defined quotient. 
On the other, it points out that algorithmic learning with tolerances is not yet well understood and requires further theoretical developments.

\paragraph{Acknowledgements}
This work has been partially funded by ANII-Agencia Nacional de Investigaci\'on e Innovaci\'on, Uruguay, under grants IA\_1\_2022\_1\_173516 and FMV\_1\_2023\_1\_175864.

\bibliography{main.bib}

\begin{thebibliography}{10}

\bibitem{Angluin:1987}
Dana Angluin.
\newblock Learning regular sets from queries and counterexamples.
\newblock {\em Inf. Comput.}, 75(2):87--106, November 1987.

\bibitem{pmlr-v57-balle16}
Borja Balle, Rémi Eyraud, Franco~M. Luque, Ariadna Quattoni, and Sicco Verwer.
\newblock Results of the sequence prediction challenge ({SPiCe}): a competition on learning the next symbol in a sequence.
\newblock In Sicco Verwer, Menno~van Zaanen, and Rick Smetsers, editors, {\em Proceedings of The 13th International Conference on Grammatical Inference}, volume~57 of {\em Proceedings of Machine Learning Research}, pages 132--136, Delft, The Netherlands, 05--07 Oct 2017. PMLR.

\bibitem{sico_pdfa_distillation_24}
R.~Baumgartner and S.E. Verwer.
\newblock {PDFA} distillation with error bound guarantees.
\newblock In Szil{\'a}rd~Zsolt Fazekas, editor, {\em Implementation and Application of Automata}, pages 51--65, Cham, 2024. Springer Nature Switzerland.

\bibitem{learnaut2024}
M.~Carrasco, F.~Mayr, S.~Yovine, J.~Kidd, M.~Iturbide, J.~da~Silva, and A.~Garat.
\newblock Analyzing constrained llm through pdfa-learning.
\newblock In {\em LearnAut 2024}, 2024.

\bibitem{tolerance76}
Ivan Chajda, Josef Niederle, and Bohdan Zelinka.
\newblock On existence conditions for compatible tolerances.
\newblock {\em Czechoslovak Mathematical Journal}, (2):304--311, 1976.

\bibitem{clark_thollard}
A.~Clark and F.~Thollard.
\newblock {PAC}-learnability of probabilistic deterministic finite state automata.
\newblock {\em J. Machine Learning Research}, 5:473--497, 2004.

\bibitem{De_la_higuera:2010}
Colin de~la Higuera.
\newblock {\em Grammatical Inference: Learning Automata and Grammars}.
\newblock Cambridge University Press, Cambridge, UK, 2010.

\bibitem{bollig_23}
Igor Khmelnitsky, Daniel Neider, Rajarshi Roy, Xuan Xie, Beno{\^{\i}}t Barbot, Benedikt Bollig, Alain Finkel, Serge Haddad, Martin Leucker, and Lina Ye.
\newblock Analysis of recurrent neural networks via property-directed verification of surrogate models.
\newblock {\em Int. J. Softw. Tools Technol. Transf.}, 25(3):341--354, 2023.

\bibitem{learnaut2022}
F.~Mayr, S.~Yovine, F.~Pan, N.~Basset, and T.~Dang.
\newblock Towards efficient active learning of {PDFA}.
\newblock In {\em LearnAut 2022}, 2022.

\bibitem{icgi2023_b}
Franz Mayr, Sergio Yovine, Mat\'ias Carrasco, Alejo Garat, Mart\'in Iturbide, Juan da~Silva, and Federico Vilensky.
\newblock Results of neural-checker toolbox in taysir 2023 competition.
\newblock In François Coste, Faissal Ouardi, and Guillaume Rabusseau, editors, {\em Proceedings of 16th edition of the International Conference on Grammatical Inference}, volume 217 of {\em Proceedings of Machine Learning Research}, pages 295--298, Rabat, Morocco, 10--13 Jul 2023. PMLR.

\bibitem{icgi2023_a}
Franz Mayr, Sergio Yovine, Mat\'ias Carrasco, Federico Pan, and Federico Vilensky.
\newblock A congruence-based approach to active automata learning from neural language models.
\newblock In François Coste, Faissal Ouardi, and Guillaume Rabusseau, editors, {\em Proceedings of 16th edition of the International Conference on Grammatical Inference}, volume 217 of {\em Proceedings of Machine Learning Research}, pages 250--264, Rabat, Morocco, 10--13 Jul 2023. PMLR.

\bibitem{mayr_yovine_2021}
Franz Mayr, Sergio Yovine, and Ramiro Visca.
\newblock Property checking with interpretable error characterization for recurrent neural networks.
\newblock {\em Machine Learning and Knowledge Extraction}, 3(1):205--227, 2021.

\bibitem{pmlr-v217-muskardin23a}
E.~Mu\v{s}kardin, M.~Tappler, and B.~K.~Aichernig.
\newblock Testing-based black-box extraction of simple models from rnns and transformers.
\newblock In {\em PMLR}, volume 217, pages 291--294, 10--13 Jul 2023.

\bibitem{IEEE:Vidal2005}
E.~Vidal, F.~Thollard, C.~de~la Higuera, F.~Casacuberta, and R.C. Carrasco.
\newblock Probabilistic finite-state machines - part {I}.
\newblock {\em IEEE PAMI}, 27(7):1013--1025, 2005.

\bibitem{weiss_WFA_learning}
G.~Weiss, Y.~Goldberg, and E.~Yahav.
\newblock Learning deterministic weighted automata with queries and counterexamples.
\newblock In {\em Adv. in Neural Information Proc. Sys.}, volume~32, 2019.

\bibitem{WeissGY24}
Gail Weiss, Yoav Goldberg, and Eran Yahav.
\newblock Extracting automata from recurrent neural networks using queries and counterexamples (extended version).
\newblock {\em Mach. Learn.}, 113(5):2877--2919, 2024.

\end{thebibliography}
\bibliographystyle{plain}

\end{document}